\documentclass[12pt]{article} 
\usepackage[utf8]{inputenc}
\usepackage{amssymb}
\usepackage{amsmath}
\usepackage{amsthm}
\usepackage{graphicx}
\usepackage{epstopdf} 
\usepackage{citesort}
\usepackage{geometry}

\usepackage{algorithm,algcompatible}
\renewcommand{\COMMENT}[2][.55\linewidth]{%
  \leavevmode\hfill\makebox[#1][l]{//~#2}}
\algnewcommand\algorithmicto{\textbf{to}}
\algnewcommand\RETURN{\State \textbf{return} }

\newtheorem{theorem}{Theorem}
\newtheorem{lemma}{Lemma}

\theoremstyle{remark} \newtheorem{remark}[theorem]{Remark}
\theoremstyle{definition} 

\title{Shifted equivalent sources and FFT acceleration for periodic scattering problems, including Wood anomalies}

\author{Oscar P. Bruno\footnote{Computing and Mathematical Sciences,
    Caltech, Pasadena, CA 91125, USA} \and Mart\'in Maas\footnote{University of Buenos Aires and CONICET, Argentina}}
\date{}

\begin{document}
\maketitle

\begin{abstract}
  This paper introduces a fast algorithm, applicable throughout the
  electromagnetic spectrum, for the numerical solution of problems of
  scattering by periodic surfaces in two-dimensional space. The
  proposed algorithm remains highly accurate and efficient for
  challenging configurations including randomly rough surfaces, deep
  corrugations, large periods, near grazing incidences, and,
  importantly, Wood-anomaly resonant frequencies. The proposed
  approach is based on use of certain ``shifted equivalent sources''
  which enable FFT acceleration of a Wood-anomaly-capable
  quasi-periodic Green function introduced recently (Bruno and
  Delourme, Jour. Computat. Phys., 262--290, 2014). The Green-function
  strategy additionally incorporates an exponentially convergent
  shifted version of the classical {\em spectral} series for the Green
  function. While the computing-cost asymptotics depend on the
  asymptotic configuration assumed, the computing costs rise at most
  linearly with the size of the problem for a number of important
  rough-surface cases we consider. In practice, single-core runs in
  computing times ranging from a fraction of a second to a few seconds
  suffice for the proposed algorithm to produce highly-accurate
  solutions in some of the most challenging contexts arising in
  applications. 
\end{abstract}



\section{Introduction\label{sec_intro}}

The problem of scattering by rough surfaces has received considerable
attention over the last few decades in view of its significant
importance from scientific and engineering viewpoints. Unfortunately,
however, the numerical solution of such problems has generally
remained quite challenging. For example, the evaluation of
rough-surface scattering at grazing angles has continued to pose
severe difficulties, as do high-frequency problems including deep
corrugations and/or large periods, and problems at certain
``Wood-anomaly'' frequencies. (As mentioned in Remarks~\ref{wood_1}
and~\ref{wood2d} below, at Wood frequencies the classical
quasi-periodic Green Function ceases to exist, and associated
Green-function summation methods such
as~\cite{ewald_linton1998,ewald_capolino2005,ewald_arens2011} become
inapplicable.)  In spite of significant progress in the general area
of scattering by periodic
surfaces~\cite{desanto98,arens_desanto2006,BrunoHaslamJOSA,Barnett2d,Barnett3d,BSTV2,BrunoDelourme},
methodologies which effectively address the various aforementioned 
difficulties for realistic configurations have remained elusive. The
present contribution proposes a new fast and accurate
integral-equation methodology which addresses these challenges in the
two-dimensional case. The method proceeds by introducing the notion of
``shifted equivalent sources'', which extends the applicability of the
FFT-based acceleration approach~\cite{BrunoKunyansky} to the context
of the Wood-anomaly capable two- and three-dimensional shifted Green
functions~\cite{BrunoDelourme,BSTV2,brunoLado}.  In the present
two-dimensional case, single-core runs in computing times ranging from
a fraction of a second to a few seconds suffice for the proposed
algorithm to produce highly-accurate solutions in some of the most
challenging contexts arising in applications---even at grazing angles
and Wood frequencies. The algorithm is additionally demonstrated for
certain extreme geometries featuring several hundred wavelengths in
period and/or depth, for which accurate solutions are obtained in
single-core runs of the order of a few minutes.

The Wood-anomaly problem has historically presented significant
difficulties. In fact, the well-known challenges that arise as
incidences approach
grazing~\cite{chenwest1995,johnson1998grazing,saillard2011} are
closely related to the appearance of associated near Wood anomalies
(Section~\ref{sec_gaussian}). As indicated in the present paper's
Remark~\ref{wood3d}, further, Wood anomalies are specially pervasive
in three-dimensional configurations, and they have therefore
significantly curtailed solution of periodic scattering problems in
that higher-dimensional context. The extension~\cite{BSTV2} of the
shifted Green function approach to three-dimensions gave rise, for the
first time, to solvers which are applicable to doubly periodic
scattering problems under Wood-frequencies in three-dimensional
space. (An alternative approach to the Wood anomaly problem for
two-dimensions was introduced in~\cite{Barnett2d}, but the
three-dimensional, bi-periodic version~\cite{Barnett3d} of that
approach is restricted to frequencies away from Wood anomalies.)  The
contribution~\cite{BSTV2} does not include an acceleration procedure,
and it can therefore prove exceedingly expensive---except when applied
to relatively simple configurations. The present paper introduces, in
the two-dimensional context, an accelerated version of the shifted
Green-function approach.  An extension of this methodology to the
three-dimensional case, which will be presented elsewhere, has been
found equally effective.


With reference to the nomenclature and concepts introduced
in~\cite{BrunoKunyansky}, in the proposed approach, a ``small'' number
of {\em free-space} equivalent-source densities are initially
computed.  Subsequent convolution of those sources with the {\em
  shifted quasi-periodic Green function}~\cite{BrunoDelourme,BSTV2,brunoLado}
produces, after necessary near-field corrections, the desired
quasi-periodic fields.  Importantly, the near-field corrections needed
in the present context are designed to account for near-field sources
inherent in the shifting strategy (Section~\ref{accel}). Additionally,
the proposed approach requires evaluation of a significantly reduced
number of quasi-periodic Green function values, as low as $O(N)$,
depending on the acceleration setup, instead of the $\mathcal{O}(N^2)$
that are generally required---thus providing highly significant
additional acceleration. The Green-function strategy is supplemented,
finally, by an exponentially convergent shifted version of the
classical spectral series for the Green function, that is used for
large portions of the Cartesian acceleration grid. Use of specialized
high order Nystr\"om quadrature rules, together with the iterative
linear algebra solver GMRES~\cite{GMRES}, complete the proposed
methodology.

This paper is organized as follows: after a few preliminaries are laid
down in Section~\ref{sec:scattering_prb}, Section~\ref{unacc}
describes the shifted Green function
method~\cite{BrunoDelourme,BSTV2,brunoLado}, and it introduces a
hybrid spatial-spectral strategy for the efficient evaluation of the
shifted Green function itself. Our high order quadrature rules and
their use of the hybrid evaluation strategy are put forth in
Section~\ref{formul}.  Section~\ref{accel} then introduces the central
concepts of this paper, namely, the shifted equivalent source method
and the associated FFT acceleration approach;
Section~\ref{overall_fast_solv} presents an algorithmic description of
the overall accelerated solution method. Section~\ref{numer}
demonstrates the new solver by means of a wide variety of
applications. After a few concluding remarks presented in
Section~\ref{concl}, theoretical questions concerning the convergence
of the proposed algorithm are taken up briefly in
Appendix~\ref{app_chap_2}.

\section{Preliminaries}
\label{sec:scattering_prb}
We consider the problem of scattering of a transverse electric
incident electromagnetic wave of the form $u^{inc}(x,y)=e^{\mathrm{i}
  (\alpha x - \beta y)}$ by a perfectly conducting periodic surface
$\Gamma = \{ \left(x,f(x) \right), x \in \mathbb{R} \}$ in
two-dimensional space, where $f$ is a smooth periodic function of
period $d$: $f(x+d) = f(x)$; the transverse magnetic case can be
treated analogously~\cite{BrunoDelourme}. Letting $\mathrm{k}^2 =
\alpha^2+\beta^2$, the scattered field $u^s$ satisfies
\begin{equation}\label{fist_eq}
\left\lbrace
\begin{array}{rlc}
  \Delta u^s + \mathrm{k}^2 u^s & = 0 &  \quad\mbox{in}\quad\Omega^{+}_f \\
  u^s & = -u^{inc} &\quad \mbox{in}\quad\Gamma,
\end{array}
\right. 
\end{equation}
where $\Omega^{+}_f = \{(x,y): y>f(x) \}$.  The incidence angle
$\theta \in (-\frac{\pi}{2},\frac{\pi}{2})$ is defined by $\alpha = \mathrm{k}
\sin(\theta)$ and $\beta = \mathrm{k} \cos(\theta)$. As is
known~\cite{Maystre}, the scattered field $u^s$ is quasi-periodic
($u^s(x+d,y) = u^s(x,y) e^{\mathrm{i} \alpha d}$) and, for all $(x,y)$ such
that $y>\max_{x\in \mathbf{r}} f(x)$, it can be expressed in terms of
a Rayleigh expansion of the form
\begin{equation} \label{eq:rayexp}
 u^{s}(x,y)= \sum_{n=-\infty}^{\infty}B_n e^{\mathrm{i}\alpha_n x+\mathrm{i}\beta_n y}.
\end{equation}
Here, $B_n\in\mathbb{C}$ are the Rayleigh coefficients, and, letting
$U$ denote the finite set of integers $n$ such that $\mathrm{k}^2 -
\alpha_n^{2} > 0$, the wavenumbers $(\alpha_n,\beta_n)$ are given by
\begin{equation}\label{alpha_beta_def}
\alpha_n := \alpha + n\frac{2\pi}{d}, \hspace{30pt}  \beta_n := \left\lbrace
  \begin{array}{ccc}
    \sqrt{\mathrm{k}^2-\alpha_n^2} & , & n\in U \\
    \mathrm{i}\sqrt{\alpha_n^2-\mathrm{k}^2} & , & n\not\in U, \\
  \end{array} \right.
\end{equation}
where the positive branch of the square root is used.

For $n\in U$, the functions $e^{\mathrm{i}\alpha_n x+\mathrm{i}\beta_n
  y}$ correspond to propagative waves.  A wavenumber $\mathrm{k}$ is
called a \textit{Wood-Rayleigh frequency} (or, for conciseness, a Wood
frequency) if for some $n\in\mathbb{Z}$ we have $\mathrm{k}^2 =
\alpha_n^{2}$, or equivalently, $\beta_n=0$. At a Wood Frequency, the
function $e^{\mathrm{i}\alpha_n x+\mathrm{i}\beta_n y} =
e^{\mathrm{i}\alpha_n x}$ becomes a grazing plane wave---that is to
say, under the present conventions, a wave that propagates parallel to
the $x$-axis. Note that at {\em grazing incidence} ($\theta =
\frac{\pi}{2}$) we have $\alpha = \mathrm{k}$ and, thus, $\beta_0 =
0$---that is, any frequency $\mathrm{k}$ becomes a Wood anomaly at
grazing incidence.

\begin{remark}\label{wood_1}
  The term ``Wood-anomaly'' relates to experimental observations by
  Wood~\cite{Wood1902} and a subsequent mathematical treatment by
  Rayleigh~\cite{Rayleigh1907} concerning conversion of propagative to
  evanescent waves as frequencies or incidence angles are changed. As
  pointed out in~\cite{maystre2012theory}, it would be more
  appropriate to refer to this phenomenon as Wood-Rayleigh anomalies
  and frequencies, but, throughout this paper, we use the Wood anomaly
  nomenclature in keeping with common
  practice~\cite{McPhedran,StewartGallaway,Barnett2d}. A brief
  discussion of historical aspects concerning this terminology can be
  found in~\cite[Remark 2.2]{BrunoDelourme}.
\end{remark}

For $n \in U$, the $n$-th order \textit{efficiency},  which is defined by
$e_n = \frac{\beta_n}{\beta}|B_n|^2$, represents the
fraction of the incident energy that is reflected in the $n$-th propagative mode. In particular, as is well known~\cite{Maystre}, for a perfectly conducting surface the (finitely many) efficiencies $e_n$ satisfy the energy balance criterion: $\sum_{n \in U} e_n = 1$. Since integral equation methods  do not enforce this relation exactly, the numerical ``energy-balance'' error 
\begin{equation}\label{energy_error}
\varepsilon = 1 - \sum_{n \in U} e_n, 
\end{equation}
is commonly used to evaluate the precision of numerical
solutions. When supplemented by checks based on convergence studies as
resolutions are increased, the resulting energy-balance error
criterion can be very useful and reliable.

Calling 
\begin{equation}\label{free_space}
G(X,Y) = \frac{\mathrm{i}}{4}H_0^1(\mathrm{k}\sqrt{X^2+Y^2}),
\end{equation}
the free-space Green function for the Helmholtz operator $\Delta +
\mathrm{k}^2$ (where $H_0^{(1)}$ denotes the first Hankel function of order
zero), the classical quasi-periodic Green function for the
problem~\eqref{fist_eq} is given by
\begin{equation} \label{eq:classic_quasiperiodic} G^q(X,Y) = \sum_{n
    \in \mathbf{Z}} e^{\mathrm{i} \alpha n d} G(X+nd,Y).
\end{equation}
The Green function~\eqref{eq:classic_quasiperiodic} also admits the
Rayleigh representation
\begin{equation} \label{quasi_per_rayleigh}
 G^{\textit{qper}}(X,Y) = \frac{\mathrm{i}}{2d} \sum_{n \in \mathcal{Z}} \frac{e^{\mathrm{i} \alpha_n X + \mathrm{i} \beta_n |Y|} }{\beta_n},
\end{equation}
a suitable modification of which can be exploited, as shown in
Section~\ref{hybrid_spectral}, to significantly accelerate the
Wood-frequency capable shifted Green function introduced in
Section~\ref{unacc}.
\begin{remark}\label{wood2d}
  It is important to note that at Wood frequencies the grazing wave
  $e^{\mathrm{i} \alpha_n X}$ in~\eqref{quasi_per_rayleigh}
  ($\beta_{n}=0$) acquires an infinite
  coefficient. Accordingly~\cite{MBV0}, at Wood frequencies the
  lattice sum \eqref{eq:classic_quasiperiodic} blows up. The shifting
  strategy introduced in Section~\ref{unacc} gives rise to
  quasi-periodic Green functions which do not suffer from these
  difficulties.
\end{remark}
\begin{remark}\label{wood3d}
   Wood frequencies (and, thus also ``near-Wood frequencies'') are
  particularly ubiquitous in the 3D case. Indeed, while in two dimensions the proximity to a Wood frequency configuration 
  is characterized by the closest distance from $\mathrm{k}^2$ to the discrete lattice $\alpha_n^2$, that is, by the quantity
\begin{equation}
 R_\mathrm{wood} = \min_{n \in \mathbf{Z}} \beta_{n} = \min \left\lbrace  \sqrt{ \mathrm{k}^2 - \left( \alpha + 2\pi n/d \right)^2 }: n \in \mathbf{Z} \right\rbrace,
\end{equation}
the corresponding expression for the 3D case (for incidence wavevector $(\alpha_1,\alpha_2,-\beta)$) is given by
\begin{equation}
 R_\mathrm{wood} = \min \left\lbrace \sqrt{ \mathrm{k}^2 - \left( \alpha_1 + 2\pi n/d_1 \right)^2 - \left( \alpha_2 + 2\pi m/d_2 \right)^2 }: (n,m) \in \mathbf{Z}^2 \right\rbrace.
\end{equation}
Therefore, in 3D, near Wood frequencies arise for all points in the
lattice $(\alpha_1,\alpha_2) + \left( \frac{2\pi}{d_1}
  \mathbf{Z}\right) \times \left(\frac{2\pi}{d_2} \mathbf{Z} \right)$
that lie on circles of radii close to $\mathrm{k}$ and are thus quite
numerous for large values of $\mathrm{k}$: for a given arbitrarily
small distance $\varepsilon$, in the 3D case there are $O(\mathrm{k})$
such frequencies within an $\varepsilon$ band around the circle of
radius $\mathrm{k}$ in the plane. For sufficiently small $\varepsilon$
the corresponding number in the 2D case is at most two.
\end{remark}

\section{Shifted Green function\label{unacc}}


As shown in~\cite{BrunoDelourme,BSTV2}, a suitable modification of the
Green function~\eqref{quasi_per_rayleigh} which does not suffer from
the difficulties mentioned in Remark~\ref{wood2d}, and which is
therefore valid throughout the spectrum, can be introduced on the
basis of a certain ``shifting'' procedure related to the method of
images.  In what follows, the construction~\cite{BrunoDelourme} of a
multipolar or ``shifted'' quasi-periodic Green function is reviewed
briefly, and a new hybrid spatial-spectral strategy for its evaluation
is presented.

\subsection{Quasi-periodic multipolar Green functions \label{sec_quasiper_multipolar}}
Rapidly decaying multipolar Green functions $G_j$ of various orders
$j$ can be obtained as linear combinations of the regular free-space
Green function $G$ with
arguments that include a number $j$ of shifts. For
example, we define a multipolar Green function of order $j=1$ by
\begin{equation}
 G_1(X,Y) = G(X,Y) - G(X,Y+h)
\end{equation}
This expression provides a Green function for the Helmholtz equation,
valid in the complement of the shifted-pole set $P_1=\{ (0,-h) \}$, which decays faster than $G$ (with order $|X|^{-\frac 32}$ instead of $|X|^{-\frac 1 2}$) as $X \to \infty$---as there results from a simple
application of the mean value theorem and the asymptotic properties of
Hankel functions~\cite{Lebedev}. 

A suitable generalization of this idea, leading to multipolar Green
functions with arbitrarily fast algebraic decay~\cite{BrunoDelourme},
results from application of the finite-difference operator
$(u_0,\dots,u_j)\to \sum_{\ell=0}^j (-1)^\ell\binom{j}{\ell}u_{\ell}$
($j\in\mathbb{N}$) that, up to a factor of $1/h^j$,
approximates the $j$-th order $Y$-derivative
operator~\cite[eq. 5.42]{Knuth}. For each non-negative integer $j$,
the resulting multipolar Green functions $G_j$ of order $j$ is thus
given by
\begin{equation}\label{shifted_green_function}
  G_j(X,Y) = \sum_{m=0}^j (-1)^m C_m^j  \, G(X,Y+mh), \quad \mbox{where} \quad C_m^j = \binom{j}{m} = \frac{j!}{m! (j-m)!}.
\end{equation}
Clearly, $G_j$ is a Green function for the Helmholtz equation in the
complement of the shifted-pole set
\begin{equation} 
  P_j = \{(X,Y)\in \mathbb{R}^2 : (X,Y) = (0,-m h)
  \mbox{ for some } m\in\mathbb{Z} \mbox{ with } 1\leq m\leq j\}.
\end{equation}
As shown in~\cite{BrunoDelourme}, further, for $Y$ bounded we have
\begin{equation}\label{shifted_green_space}
  G_j(X,Y)\sim |X|^{-q} \mbox{ as } X \to\infty, \quad \mbox{with} \quad q = \frac{1}{2} + \left\lfloor{\frac{j+1}{2}}\right\rfloor,
\end{equation}
where $\left\lfloor{x}\right\rfloor$ denotes the largest integer less than or
equal to $x$.

For sufficiently large values of $j$, the spatial lattice sum
\begin{equation}
\label{shifted_periodic_green_function}
 \tilde G_j^{\textit{qper}}(X,Y) = \sum_{n=-\infty}^{\infty} e^{-\mathrm{i}\alpha n d} G_j(X+nd,Y)
\end{equation}
provides a rapidly (algebraically) convergent quasi-periodic Green function series defined for all $(X,Y)$
outside the periodic shifted-pole lattice
\begin{equation} 
  P_j^{\textit{qper}} = \{(X,Y)\in \mathbb{R}^2 : (X,Y) = (nd,-m h)
  \mbox{ for some }n,m\in\mathbb{Z} \mbox{ with } 1\leq m\leq j\}.
\end{equation}
The Rayleigh expansion of $\tilde G_j^{\textit{qper}}$, further, can be readily
obtained by applying equation~\eqref{quasi_per_rayleigh}; the result
is
\begin{equation}\label{shifted_quasi_per_rayleigh}
 \tilde G_j^{\textit{qper}}(X,Y) = \sum_{n=-\infty}^{\infty} \frac{\mathrm{i}}{2d\beta_n} e^{\mathrm{i}\alpha_n X} \left( \sum_{m=0}^j (-1)^m C_m^j e^{\mathrm{i} \beta_n |Y + m h|} \right) \qquad \mbox{for } Y \not=-mh, \quad 0\le m \le j.
\end{equation}
And, using the identity $\sum_{m=0}^j (-1)^m C_m^j e^{\mathrm{i}
  \beta_n (Y + m h)} = e^{\mathrm{i} \beta_n Y} (1-e^{\mathrm{i} \beta_n h})^j $
there results
%
\begin{equation}\label{shifted_periodic_green_function_spectral}
 \tilde G_j^{\textit{qper}}(X,Y) = \sum_{n=-\infty}^{\infty} \frac{\mathrm{i}}{2d\beta_n} (1-e^{\mathrm{i} \beta_n h})^j e^{\mathrm{i}\alpha_n X + \mathrm{i} \beta_n Y} \qquad \mbox{for } Y > 0.
\end{equation}

As anticipated, no problematic infinities occur in the Rayleigh
expansion of $\tilde G_j^{\textit{qper}}$, even at Wood anomalies
($\beta_n=0$), for any $j\geq 1$. The shifting procedure has thus
resulted in rapidly-convergent spatial representations of various
orders (equations~\eqref{shifted_green_space}
and~\eqref{shifted_periodic_green_function}) as well as spectral
representations which do not contain infinities
(equation~\eqref{shifted_periodic_green_function_spectral}).

An issue does arise from the shifting method which requires attention:
the shifting procedure cancels certain Rayleigh modes for $Y>0$ and
thereby affects the ability of the Green function to represent general
fields.  In detail, the coefficient $(1-e^{\mathrm{i} \beta_n h})^j
\beta_n^{-1}$ in the
series~\eqref{shifted_periodic_green_function_spectral} vanishes if
either $\beta_n = 0 $ (Wood anomaly) and $j\geq 1$, or if $\beta_n h$
equals an integer multiple of $2\pi$.  As in~\cite{BrunoDelourme}, we
address this difficulty by simply adding to $\tilde
G_j^{\textit{qper}}$ the missing modes. In fact, in a numerical
implementation it is beneficial to incorporate corrections containing
not only resonant modes, but also {\em nearly} resonant modes. Thus,
using a sufficiently small number $\eta$ and defining the
$\eta$-dependent completion function
\begin{equation}\label{M_eta}
  M^\eta(X,Y)  = \sum_{n\in U^\eta} e^{\mathrm{i}\alpha_n X + \mathrm{i}\beta_n Y}, \qquad U^\eta =  \left\{  n \in \mathbb{Z}: |(1-e^{\mathrm{i} \beta_n h})^j \beta_n^{-1}| < \eta \right\},
\end{equation}
(where for $\beta_n = 0$ the quotient $|(1-e^{\mathrm{i} \beta_n h})^j
\beta_n^{-1}|$ is interpreted as the corresponding limit as $\beta_n
\to 0$), a {\em complete} version of the shifted Green function is
given by
\begin{equation}\label{complete}
G_j^{\textit{qper}}(X,Y) = \tilde G_j^{\textit{qper}}(X,Y) + M^\eta(X,Y)
\end{equation}
for $(X,Y)$ outside the set $P_j^{\textit{qper}}$. 

\begin{remark}\label{rem_gtilde}
  The following section presents an algorithm which, relying on both
  equations~\eqref{shifted_periodic_green_function}
  and~\eqref{shifted_quasi_per_rayleigh}, rapidly evaluates the Green
  function $\tilde G_j^{\textit{qper}}$.  Section~\ref{formul}
  presents integral equation formulations based on separate use of the
  functions $\tilde G_j^{\textit{qper}}$ and $M^\eta$, that avoids a
  minor difficulty (addressed in~\cite[Remark 4.8]{BrunoDelourme})
  related to the direct use of the Green function
  $G_j^{\textit{qper}}$ defined in~\eqref{complete}.
\end{remark}

\subsection{Hybrid spatial-spectral evaluation of $\tilde
  G_j^{\textit{qper}}$}\label{hybrid_spectral}
Equation~\eqref{shifted_periodic_green_function_spectral} provides a
very useful expression for evaluation of $\tilde G_j^{\textit{qper}}$ for $Y>0$ at
all frequencies, including Wood anomalies---since, for such 
values of $Y$, this series converges exponentially fast. Interestingly,
further, the related expression~\eqref{shifted_quasi_per_rayleigh} can
also be used, again, with exponentially fast convergence, including
Wood anomalies, for all values of $Y$ sufficiently far from the set $\{Y=-m
h: \; 0\leq m\leq j \}$. The latter expression thus provides a greatly
advantageous alternative to direct summation of the
series~\eqref{shifted_periodic_green_function} for a majority (but not
not the totality) of points $(X,Y)$ relevant in a given quasi-periodic
scattering problem.

The exponential convergence of~\eqref{shifted_quasi_per_rayleigh} is clear by inspection. 
To see that~\eqref{shifted_quasi_per_rayleigh} is well defined
at and around Wood anomalies it suffices to substitute the sum in $m$ in equation~\eqref{shifted_quasi_per_rayleigh} by the expression
\begin{equation}\label{betan_taylor}
 \sum_{m=0}^j (-1)^m C_m^j \frac{e^{\mathrm{i} \beta_n |Y + m h|}}{\beta_n} =  e^{\mathrm{i} \beta_n Y} \frac{(1-e^{\mathrm{i} \beta_n h})^j}{\beta_n} - \sum_{\substack{0\le m\le j \\ m < -Y/h}} (-1)^m C_m^j \frac{ e^{\mathrm{i} \beta_n (Y + m h)} - e^{-\mathrm{i} \beta_n (Y + m h)}}{\beta_n}. 
\end{equation}
where, once again, the values of the quotients containing $\beta_n$
denominators at $\beta_n=0$ are interpreted as the corresponding
$\beta_n \to 0$ limits. 

A strategy guiding the selection of the values $Y$ for which the spectral
series~\eqref{shifted_quasi_per_rayleigh} is used instead of the spatial 
series~\eqref{shifted_periodic_green_function} can be devised on the basis of the relation 
\begin{equation}\label{beta_n_estimate}
 \beta_n = \mathrm{k}\sqrt{1-(\sin(\theta)+\frac{\lambda}{d}n)^2} \; \approx \; \mathrm{\mathrm{i}}\mathrm{k} \frac{\lambda}{d}n + \mathcal{O}(1) \; = \; \frac{2n\pi}{d}\mathrm{i} + \mathcal{O}(1).
\end{equation}
Indeed, the estimate 
\begin{equation}\label{spectral_exp_estimate}
 \left| e^{\mathrm{i}\beta_n|Y+m h|} \right| < C e^{- 2 n \pi\frac{\delta}{d}}, \quad (|Y+mh|>\delta>0)
\end{equation}
shows that, for $|Y+mh|>\delta>0$, the spectral representation~\eqref{shifted_quasi_per_rayleigh} 
converges like a geometric series of ratio $e^{- 2 \pi\frac{\delta}{d}} <1$---with fast convergence for values of $\frac{\delta}{d}$ sufficiently far from zero.

\section{Hybrid, high-order Nystr\"om solver throughout the spectrum}\label{formul}
\subsection{Integral equation formulation\label{formulation}}
The Green functions $G_j^{\textit{qper}}$ presented in
Section~\ref{unacc} (equation~\eqref{complete}) can be used to devise
an integral equation formulation for problem~\eqref{fist_eq} which
remains valid at Wood Anomalies~\cite{BrunoDelourme}. As indicated in
Remark~\ref{rem_gtilde}, however, we proceed in a slightly different
manner. Letting $\nu(x')$ denote the normal to the curve $\Gamma$ at
the point $(x',f(x'))$ and $ds'$ denote the element of length on
$\Gamma$ at $(x',f(x'))$, we express the scattered field $u^{scat}$
in~\eqref{fist_eq}, for all $(x,y)\in\Omega^{+}_f$, as a multipolar
double layer potential plus a potential with kernel $M^\eta$:
\begin{equation}
\label{eq:int_eq_dirichlet}
u^{scat}(x,y) = \int_0^d \left( \nu(x')\cdot \nabla_{(x',y')} \tilde G_j^{\textit{qper}}(x-x',y-y')\big|_{y'=f(x')} + M^\eta(x-x',y-f(x')) \right) \mu(x') ds'.
\end{equation}
Defining the normal-derivative operator $\partial_{\nu'}$, whose action on a given function $K:
\mathbb{R} \times \mathbb{R} \to \mathbb{C}$ is given by
\begin{equation}\label{K_def}
\partial_{\nu'} K (x,x')  =  \left[ \nu(x')\cdot \nabla_{(x',y')}K  (x-x',y-y')\right]_{y=f(x),y'=f(x')},
\end{equation}
and, letting $D$ denote the integral operator
\begin{equation}\label{per_doublelayer_op}
 D[\mu](x) = \int_0^d \left( \partial_{\nu'} \tilde G_j^{\textit{qper}}(x,x') + M^\eta(x-x',f(x)-f(x')) \right) \mu(x') ds', \quad x\in[0,d],
\end{equation}
it follows that $\mu$ satisfies the integral equation
\begin{equation}\label{eq_per_doublelayer}
  \frac{1}{2} \mu(x) + D[\mu](x)   = -u^{inc}(x) \quad \mbox{for}\quad x\in[0,d]. 
\end{equation}
We may also write 
\begin{equation}\label{d_tilde_M}
  D[\mu] = \tilde D[\mu] + D_M[\mu]
\end{equation} 
where
\begin{align}
  \tilde D[\mu](x) = & \int_0^d \partial_{\nu'} \tilde G_j^{\textit{qper}}(x,x') \mu(x') ds' \quad \mbox{and} \label{per_doublelayer_2} \\
  D_M[\mu](x) = & \int_0^d M^\eta(x-x',f(x)-f(x'))\mu(x') ds'. \label{per_doublelayer_DM}
\end{align}
It is easy to check~\cite{BrunoDelourme}, finally, that the operator
$\tilde D$ can be expressed as the infinite integral
\begin{equation}\label{per_doublelayer_op_2}
 \tilde D[\mu](x) = \int_{-\infty}^{+\infty} \partial_{\nu'} G_j(x,x') \mu(x') ds_\Gamma(x'),
\end{equation}
where $\mu$ is extended to all of $\mathbb{R}$ by
$\alpha$-quasi-periodicity:
\begin{equation}\label{mu_quasiper}
\mu(x+d)=\mu(x)e^{\mathrm{i}\alpha d}.
\end{equation}

The proposed fast iterative Nystr\"om solver for
equation~\eqref{eq_per_doublelayer} is based on use of an equispaced
discretization of the periodicity interval $[0,d]$, an associated
quadrature rule, and an FFT-based acceleration method. The underlying
high-order quadrature rule, which is closely related to the one used
in~\cite[Sect. 5]{BrunoDelourme}, but which incorporates a
highly-efficient hybrid spatial-spectral approach for the evaluation
of the Green function, is detailed in Section~\ref{quadrature}. On the
basis of this quadrature rule alone, an unaccelerated Nystr\"om solver
for equation~\eqref{eq_per_doublelayer} is presented in
Section~\ref{Overall}; a discussion concerning the convergence of this
algorithm is put forth in Appendix~\ref{app_chap_2}.  The proposed
acceleration technique and resulting overall accelerated solver are
presented in Section~\ref{accel}.

\subsection{High-order quadrature for the incomplete operator
  $\tilde{D}$ \label{quadrature}}

In the proposed Nystr\"om approach, the smooth windowing function 
\begin{equation}
\label{S_def}
S_{\gamma,a}(x) = 
\left\{
  \begin{array}{ccc}
  1 	& \text{if } |x| \le \gamma, & \\
  \exp \left(\frac{2e^{-1/u}}{u-1}\right) & \text{if } \gamma < |x| < a, & u = \frac{|x|-\gamma}{a-\gamma}, \\
  0 	& \text{if } |x| \ge a, & \\
  \end{array} 
\right.  
\end{equation}
(see Figure~\ref{POU}) is used to decompose the operator $\tilde{D}$ in
equation~\eqref{per_doublelayer_op_2} as a sum $\tilde{D} = \tilde D_{\mathrm{reg}} + \tilde D_{\mathrm{sing}}$
of \textit{regular} and \textit{singular} contributions $\tilde D_{\mathrm{reg}}$ and $\tilde D_{\mathrm{sing}}$, given by
\begin{equation}
 \tilde D_{\mathrm{reg}} [\mu](x) = \int_{-\infty}^{+\infty} \partial_{\nu'} G_j(x,x') (1-S_{\gamma,a}^f(x,x')) \mu(x') ds'
\end{equation}
and
\begin{equation}\label{dsing_def}
 \tilde D_{\mathrm{sing}}[\mu](x) = \int_{x-a}^{x+a} \partial_{\nu'} G_j(x,x') S_{\gamma,a}^f(x,x') \mu(x')  ds'.
\end{equation}
where we have defined
\begin{equation}
S_{\gamma,a}^f(x,x') = S_{\gamma,a}\left(\sqrt{(x-x')^2 + (f(x)-f(x'))^2}\right).
\end{equation}

\begin{remark}\label{selec_a}
  The parameter $a$ is selected so as to appropriately isolate the
  logarithmic singularity. For definiteness, throughout this chapter it
  is assumed the relation $a<d$ is satisfied.
\end{remark}

\begin{figure}[ht!]
 \centering
 \includegraphics[scale=0.5]{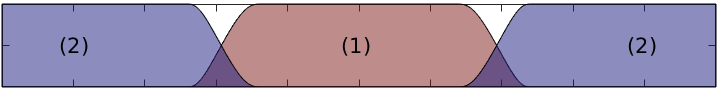}
 \caption{Partition of Unity functions
   $S_{\gamma,a}(x)$ and
   $1-S_{\gamma,a}(x)$, labeled (1) and (2),
   respectively.}
\label{POU}
\end{figure}

To derive quadrature rules for the operators $\tilde D_{\mathrm{reg}}$
and $\tilde D_{\mathrm{sing}}$ we consider an equispaced
discretization mesh $\{ x_\ell \}_{\ell=-\infty}^\infty$, of mesh-size
$\Delta x = (x_{\ell+1} - x_\ell)$, for the complete real line, which
is additionally assumed to satisfy $x_0=0$ and $x_{N}=d$ for a certain
integer $N>0$. The corresponding numerical approximations of the
values $\mu(x_\ell)$ ($1\leq \ell\leq N$) will be denoted by
$\mu_\ell$; in view of~\eqref{mu_quasiper} the quantities $\mu_\ell$
are extended to all $\ell \in \mathbb{Z}$ by quasi-periodicity:
\begin{equation}\label{quasiper_disc}
  \mu_{(\ell +p N) } = \mu_{\ell}\ e^{\mathrm{i}\alpha pd}\qquad \ell = 1,\dots, N\quad p\in\mathbb{Z}.
\end{equation}

\subsubsection{Discretization of the operator $\tilde D_{\mathrm{sing}}$\label{sing_op}}

To discretize the operator $\tilde D_{\mathrm{sing}}$ we employ the
Martensen-Kussmaul (MK) splitting~\cite{ColtonKress_InverseAcoustic}
of the Hankel function $H_1^1$ into logarithmic and smooth
contributions. Following~\cite[Secs. 5.1-5.2]{BrunoDelourme} we thus
obtain the decomposition
\begin{equation}\label{descomp}
 \partial_{\nu'} G_j(x,x') = K_s(x,x')\ln\left[4\sin^2\left(\frac{\pi}{a}(x-x') \right)\right] + K_r(x,x')
\end{equation}
where the smooth kernels $K_s$ and $K_r$ are given by
\begin{equation}\label{ks_def}
 K_s(x,x') = \frac{\mathrm{k}}{4\pi} \frac{f(x')(x-x') - (f(x')-f(x))}{\sqrt{(x-x')^2+(f(x)-f(x'))^2}} J_1(\mathrm{k}\sqrt{(x-x')^2+(f(x)-f(x'))^2})
\end{equation}
and 
\begin{equation}\label{kr_def}
  K_r(x,x') = \partial_{\nu'} G_j(x,x') - K_s(x,x')\ln\left[4\sin^2\left(\frac{\pi}{a}(x-x') \right)\right].
\end{equation}
Replacing~\eqref{descomp} into~\eqref{dsing_def} we obtain
$\tilde D_{\mathrm{sing}} =
\tilde D_{\mathrm{sing}}^\mathrm{log}+\tilde D_{\mathrm{sing}}^\mathrm{trap}$ where
\begin{align}\label{local_twointeg}
  \tilde D_{\mathrm{sing}}^\mathrm{log} &=   \int_{x-a}^{x+a} \ln\left[4\sin^2\left(\frac{\pi}{a}(x-x') \right)\right] K_s(x,x') S_{\gamma,a}^f(x,x') \mu(x')  ds' \quad\mbox{and}\\
  \tilde D_{\mathrm{sing}}^\mathrm{trap} &= \int_{x-a}^{x+a} K_r(x,x')
  S_{\gamma,a}^f(x,x') \mu(x') ds'.
\end{align}
The operator $\tilde D_{\mathrm{sing}}^\mathrm{log}$ contains the
logarithmic singularity; the operator $\tilde
D_{\mathrm{sing}}^\mathrm{trap}$ on the other hand, may be
approximated accurately by means of the trapezoidal rule. 

Given that $S_{\gamma,a}^f(x,x')$ vanishes smoothly at $x'=x\pm a$
together with all of its derivatives, we can obtain high-order
quadratures for each of these integrals on the basis of the equispaced
discretization $\{ x_\ell \}$ ($\ell \in \mathbb{Z}$) and the Fourier
expansions of the smooth factor $K_s(x,x') S_{\gamma,a}^f(x,x')
\mu(x')$. Indeed, utilizing the aforementioned discrete approximations
$\mu_\ell$ (where $\ell$ may lie outside the range $1\leq \ell \leq
N$), relying on certain explicitly-computable Fourier-based weights
$R_{i\ell} $ (which can be computed for general $a$ by following the
procedure used in~\cite[Sec. 5.2]{BrunoDelourme} for the particular
case in which $a$ equals a half period $d/2$), and appropriately
accounting for certain near-singular terms in the kernel $K_r$ by
Fourier interpolation of $\mu(x')S_{\gamma,a}^f(x,x')$ (as detailed
in~\cite[Sec. 5.3]{BrunoDelourme}), a numerical-quadrature
approximation
\begin{equation}\label{local_mat}
 \tilde D^{\Delta x}_{\mathrm{sing}}[\mu_1,\dots,\mu_{N}](x_i) = \sum_{ \ell\in L_i^a } R_{i\ell} K_s(x_i,x_\ell) S_{\gamma,a}^f(x_i,x_\ell)\mu_\ell + \sum_{ \ell\in L_i^a } W_{i\ell} K_r(x_i,x_\ell) S_{\gamma,a}^f(x_i,x_\ell)\mu_\ell
\end{equation}
of $\tilde D_{\mathrm{sing}}[\mu](x_i)$ is obtained. Here $L_{i}^a :
\{ \ell : |x_\ell-x_i|\le a \}$, and, for values of $\ell > N$ and
$\ell < 1$, $\mu_\ell$ is given by~\eqref{quasiper_disc}.

\subsubsection{Discretization of the operator $\tilde D_{\mathrm{reg}}$ \label{disc_Dreg}}
The windowing function $S_{\gamma,a}$ (with ``relatively small''
values of $a$) was used in the previous section to discriminate
between singular and regular contributions $\tilde D_{\mathrm{sing}}$
and $\tilde D_{\mathrm{reg}}$ to the operator $\tilde D$. A new
windowing function $S_{cA,A}(x-x')$ (intended for use with ``large''
values of $A$) is now introduced to smoothly truncate the infinite
integral that defines the operator $\tilde D_{\mathrm{reg}}$: the
truncated operator is defined by
\begin{equation}\label{na_A}
 \tilde D_{\mathrm{reg}}^A[\mu](x) = \int_{x-A}^{x+A}  \partial_{\nu'} G_j(x,x') (1-S_{\gamma,a}^f(x,x')) \mu(x') S_{cA,A}(x-x') ds'.
\end{equation}
Defining the windowed Green function by
\begin{equation}\label{Gq_limit}
  \tilde  G^{q,A}_j(X,Y) = \sum_{p=-\infty}^{\infty} G_j(X+dp,Y) S_{cA,A}(X+dp),
\end{equation}
the truncated operator $\tilde D_{\mathrm{reg}}^{A}$ can also be
expressed in the form
\begin{equation}
\tilde D_{\mathrm{reg}}^{A}[\mu](x) = \int_{x-d/2}^{x+d/2}  \partial_{\nu'} \tilde G^{q,A}_j(x,x') (1-S_{\gamma,a}^f(x,x')) \mu(x')ds'.
\end{equation}
On account of the smoothness of the integrand in~\eqref{na_A}, and the
fact that it vanishes identically outside $[x-A,x+A]$, the
integral~\eqref{na_A} is approximated with superalgebraic order of
integration accuracy by the discrete trapezoidal rule expression
\begin{equation}\label{DAdx_def}
 \tilde D_{\mathrm{reg}}^{A,\Delta x}[\mu_1,\dots,\mu_N](x_i) = \sum_{\ell=-\infty}^{\infty}  \partial_{\nu'} G_j(x_i,x_\ell) S_{cA,A}(x_i-x_\ell)(1-S_{\gamma,a}^f(x_i,x_\ell)) \mu_\ell (\Delta s)_\ell
\end{equation}
where $(\Delta s)_\ell$ denotes the discrete surface element $\Delta x
\sqrt{1+f(x_\ell)^2}$, and with $\mu_\ell$ replaced by $\mu(x_\ell)$
($\ell = 1,\dots,N$); see also~\eqref{mu_quasiper}
and~\eqref{quasiper_disc}. 
\begin{remark}\label{rem_trapezoidal_Dreg}
  The claimed superalgebraic integration accuracy of the right-hand
  expression in~\eqref{DAdx_def} for a fixed value of $A$ follows from
  the well known trapezoidal-rule integration-accuracy result for
  smooth periodic function integrated over their
  period~\cite{Kress})---since the restriction of the integrand to
  $[x-A,x+A]$ can be extended to all of $\mathbb{R}$ as a smooth and
  periodic function $F^{A,x} = F^{A,x}(x')$ of period $2A$.
\end{remark}

An analysis of the smooth truncation procedure,
namely, of the convergence of $\tilde D_{\mathrm{reg}}^A$ to $\tilde D_{\mathrm{reg}}$ as
$A\to\infty$, is easily established on the basis of the convergence analysis~\cite{BrunoDelourme,BSTV1,BSTV2}
for the windowed-Green-function~\eqref{Gq_limit} to the regular shifted
series~\eqref{shifted_periodic_green_function}
\begin{equation}\label{Gq_limit_2}
 \tilde G^{\textit{qper}}_j(x,y) = \lim_{A\to\infty} \tilde G^{q,A}_j(x,y).
\end{equation}
The overall error resulting from the combined use of smooth truncation
and trapezoidal discretization is discussed in
Section~\ref{app_chap_2}. In particular, Lemma~\ref{lemma_app_1} in
that appendix provides an error estimate that shows that the
superalgebraic order of trapezoidal integration accuracy is also
uniform with respect to $A$.



Clearly, the numerical method embodied in equations~\eqref{local_mat}
and~\eqref{DAdx_def} provides a high-order strategy for the evaluation
of the operator $\tilde D$ in
equation~\eqref{per_doublelayer_op_2}. As shown in the following
section, a hybrid spatial/spectral Green-function evaluation strategy
can be used to significantly decrease the costs associated with
evaluation of the discrete operator in
equation~\eqref{DAdx_def}. While this strategy suffices in many cases,
when used in conjunction with the FFT acceleration method introduced
in Section~\ref{accel} a solver results which, as mentioned in the
introduction, enables treatment of challenging rough-surface
scattering problems.


\subsubsection{Spatial/Spectral hybridization \label{hybridization}}
To obtain a hybrid strategy we express~\eqref{DAdx_def} in terms of
the function $\tilde G^{\textit{qper}}_j$ which we then evaluate by means of
either~\eqref{shifted_quasi_per_rayleigh} or~\eqref{Gq_limit_2},
whichever is preferable for each pair $(x_i,x_\ell)$. Taking limit as
$A\to\infty$ in~\eqref{shifted_quasi_per_rayleigh} we obtain the
limiting discrete operator
\begin{equation} \label{D_ns_inf} \tilde D_{\mathrm{reg}}^{\Delta
    x}[\mu_1,\dots,\mu_N](x_i) = \sum_{\ell=-\infty}^{\infty}
  \mu_\ell (1-S_{\gamma,a}^f(x_i,x_\ell)) \partial_{\nu'}
  G_j(x_i,x_\ell) (\Delta s)_\ell.
\end{equation}
Writing, for every $\ell\in\mathbb{Z}$, $x_\ell = x_k - dp$ for a
unique integers $k$ and $p$ ($1\le k\le N$), exploiting the
periodicity of the function $f$ and the $\alpha$-quasi-periodicity of
$\mu$, using~\eqref{Gq_limit} and \eqref{Gq_limit_2}, and taking into
account Remark~\ref{selec_a}, we obtain the alternative expression
\begin{equation} \label{D_ns} \tilde D_{\mathrm{reg}}^{\Delta
    x}[\mu_1,\dots,\mu_N](x_i) = \sum_{m=1}^{N}
   \partial_{\nu'} \tilde  G^{\textit{qper},\star}_j(x_i,x_k)\mu_k (\Delta s)_k,
\end{equation}
where we have set
\begin{equation}\label{G_jq_s}
 \tilde  G^{\textit{qper},\star}_j(X,Y) = \tilde G^{\textit{qper}}_j(X,Y) - \sum_{p=-1}^1 G_j(X+dp,Y) e^{-\mathrm{i}\alpha dp} S_{\gamma,a}(X+dp).
\end{equation}
Clearly, $\tilde G^{\textit{qper},\star}_j$ is a smooth function that results from
subtraction from $\tilde G^{\textit{qper}}_j(X,Y)$ of (windowed
versions of) the nearest interactions (modulo the period). 

The expression~\eqref{D_ns} relies, via~\eqref{G_jq_s}, on the
evaluation of the exact quasi-periodic Green function $\tilde
G^{\textit{qper}}_j(X,Y)$. For a given point $(X,Y)$ this function can
be evaluated by either a spectral or a spatial approach: use of the
spectral series as described in Section~\ref{hybrid_spectral} is
preferable for values of $Y$ sufficiently far from the set $\{Y=-m h:
\; 0\leq m\leq j \}$, while, in view of the fast
convergence~\cite{BrunoDelourme,BSTV1,BSTV2} of \eqref{Gq_limit_2},
for other values of $Y$ the spatial expansion~\eqref{Gq_limit} with a
sufficiently large value of $A$ can be more advantageous. (Note that
if the grating is deep enough, then $f(x)$ could be far from $f(x')$
even if $x$ is relatively close to $x'$. The exponentially convergent
spectral approach could provide the most efficient alternative in such
cases.)

\subsection{Overall discretization and (unaccelerated) solution 
  of equation~\eqref{eq_per_doublelayer} \label{Overall}}
Taking into account equation~\eqref{d_tilde_M} in conjunction with the
Green function evaluation and discretization strategies presented in
Section~\eqref{quadrature} for the operator $\tilde D$, a full
discretization for the complete operator $D$
in~\eqref{per_doublelayer_op} can now be obtained easily: an efficient
discretization of the remaining operator $D_M$
in~\eqref{per_doublelayer_DM}, whose kernel is given by
equation~\eqref{M_eta}, can be produced via a direct application of
the trapezoidal rule. Separating the variables $X$ and $Y$ in the
exponentials $e^{\mathrm{i}\alpha_n X + \mathrm{i}\beta_n Y}$, further, the resulting
discrete operator may be expressed in the form
\begin{equation}\label{D_M_deltax}
 D_M^{\Delta x}[\mu_1,\dots,\mu_N](x_i) = \sum_{n\in U^\eta} e^{\mathrm{i}\alpha_n x_i} \left( \sum_{\ell=1}^N e^{\mathrm{i}\beta_n f(x_\ell)} \mu_\ell (\Delta s)_\ell \right).
\end{equation}
Letting
\begin{equation}\label{D_deltax}
 D^{\Delta x} = \tilde D^{\Delta x}_{\mathrm{sing}} + \tilde D^{\Delta x}_{\mathrm{reg}} + D^{\Delta x}_M
\end{equation}
we thus obtain the desired discrete version
\begin{equation}\label{discrete_eq}
\left( \frac{1}{2}I +   D^{\Delta x} \right)[\mu_1,\dots,\mu_N](x_i) = -u^{inc}(x_i)
\end{equation}
of equation~\eqref{eq_per_doublelayer}.

As mentioned in Section~\ref{sec_intro}, the proposed method relies on use
of an iterative linear algebra solver such as GMRES~\cite{GMRES}. The
necessary evaluation of the action of the discrete operator $D^{\Delta
  x}$ is accomplished, in the direct (unaccelerated) implementation
considered in this section, via straightforward applications of the
corresponding expressions~\eqref{local_mat},~\eqref{D_ns}
and~\eqref{D_M_deltax} for the operators $\tilde D^{\Delta
  x}_{\mathrm{sing}}$, $\tilde D^{\Delta x}_{\mathrm{reg}}$ and
$D^{\Delta x}_M$, respectively. This completes the proposed
unaccelerated iterative solver for equation~\eqref{eq_per_doublelayer}.

The computational cost required by the various components of this
solver can be estimated as follows.
\begin{enumerate}
\item The application of the \textit{local} operator $\tilde D^{\Delta
    x}_{\mathrm{sing}}$ requires $\mathcal{O}(N)$ arithmetic
  operations, the vast majority of which are those associated with
  evaluation of the multipolar Green function $G_j$.
\item $D^{\Delta x}_M$, in turn, requires $\mathcal{O}(N)$
  operations, including the computation of a number $\mathcal{O}(N)$ of
  values of exponential functions.
\item\label{G_jq_eval} The operator $\tilde D^{\Delta
    x}_{\mathrm{reg}}$, finally, requires $\mathcal{O}(N^2)$ arithmetic
  operations, including the significant cost associated with the
  evaluation of $\mathcal{O}(N^2)$ values of the
  shifted-quasi-periodic Green function $\tilde G^{\textit{qper}}_j$.
\end{enumerate}
Clearly, the cost mentioned in point~\ref{G_jq_eval} above represents
the most significant component of the cost associated of the
evaluation of $D^{\Delta x}$. Thus, although highly accurate, the
direct $\mathcal{O}(N^2)$-cost strategy outlined above for the
evaluation of $D^{\Delta x }$ can pose a significant computational
burden for problems which, as a result of high-frequency and/or
complex geometries, require use of large numbers $N$ of unknowns. A
strategy is presented in the next section which, on the basis of
equivalent sources and Fast Fourier transforms leads to significant
reductions in the cost of the evaluation of this operator, and,
therefore, in the overall cost of the solution method.

\section{Shifted Equivalent-Source Acceleration}\label{accel}
The most significant portion of the computational cost associated with
the strategy described in the previous section concerns the evaluation
of the discrete operator $\tilde D_\mathrm{reg}^{\Delta x}$ in
equation~\eqref{D_ns_inf}.  The present section introduces an acceleration
method for the evaluation of that operator which, incorporating an
FFT-based algorithm that is applicable throughout the spectrum,
reduces very significantly the number of necessary evaluations of the
periodic Green function $\tilde G_j^{\textit{qper}}$, with corresponding reductions
in the cost of the overall approach. A degree of
familiarity with the acceleration methodology introduced
in~\cite{BrunoKunyansky} could be helpful in a first reading of this section.

Central to the contribution~\cite{BrunoKunyansky} is the introduction
of ``monopole and dipole'' representations and an associated notion of
``adjacency'' that, in modified forms, are used in the present
algorithm as well. In order to extend the applicability of the
method~\cite{BrunoKunyansky} to the context of this paper, the present
Section~\ref{accel} introduces certain ``shifted equivalent source''
representations and a corresponding validity-ensuring notion of
``adjacency''. The geometrical structure that underlies the approach
as well an outline of the reminder of Section~\ref{accel} are
presented in Section~\ref{sec_shifted_eqs}.

\subsection{Geometric setup\label{sec_shifted_eqs}}

In order to incorporate equivalent sources, the algorithm utilizes a
``reference periodicity domain'' $\Omega_\textit{per} = [0,d) \times
[h_{\mathrm{min}},h_{\mathrm{max}})$, where $h_{\mathrm{min}}$ and
$h_{\mathrm{max}}$ are selected so as to satisfy $[\min(f),\max(f)]
\subset [h_{\mathrm{min}},h_{\mathrm{max}}]$. The domain
$\Omega_\textit{per}$ is subsequently partitioned in a number
$n_{\mathrm{cell}}=n_x n_y$ of mutually disjoint square cells
$c^q$---whose side $L$, we assume, satisfies
\begin{equation}\label{L_constr}
d=n_xL \qquad \mbox{and} \qquad (h_{\mathrm{max}}-h_{\mathrm{min}})=n_y L
\end{equation}
 for certain positive
integers $n_x$ and $n_y$. We additionally denote by $\Omega_\infty =
(-\infty,+\infty) \times [h_{\mathrm{min}},h_{\mathrm{max}}]$; clearly
$\Omega_\infty$ domain that is similarly partitioned into (an infinite
number of) cells $c^q$ ($q \in \mathbb{Z}$):
\begin{equation}\label{Omegas}
  \Omega_\mathrm{per} = \bigcup_{q=1}^{n_{\mathrm{cell}}} c^q \qquad \mbox{and} \qquad   \Omega_\infty  = \bigcup_{q=-\infty}^\infty c^q = \bigcup_{n=-\infty}^\infty \Big(\Omega_\mathrm{per} +n d \Big).
\end{equation} 
\begin{remark}\label{non-resonant}
  It is additionally assumed that the side $L$ of the accelerator
  cells $c^q$ is selected in such a way that these cells are not
  resonant for the given wavenumber $\mathrm{k}$---that is to say,
  that $-\mathrm{k}^2$ is not a Dirichlet eigenvalue for the Laplace
  operator in the cells $c^q$. This is a requirement in the plane-wave
  Dirichlet-problem solver described in Section~\ref{planewave_rec}.
  Clearly, values of the parameters $L$, $h_{\mathrm{max}}$,
  $h_{\mathrm{min}}$, $n_x$ and $n_y$ meeting this constraint as well
  as~\eqref{L_constr} can be found easily. Finally, the parameter $L$
  is chosen so as to minimize the overall computing cost, while
  meeting a prescribed accuracy tolerance. In all cases considered in
  this chapter values of $L$ in the range between one and four
  wavelengths were used.
\end{remark}
\begin{remark}\label{cell_def}
  In order to avoid cell intersections, throughout this chapter the
  cells $c^q$ are assumed to include the top and right sides, but not
  to include the bottom and left sides. In other words, it is assumed
  that each cell $c^q$ can be expressed in the form $c^q = (a^q_1,
  b^q_1]\times (a^q_2, b^q_2]$ for certain real numbers $a^q_1$,
  $b^q_1$, $a^q_2$ and $b^q_2$.
\end{remark}
\begin{remark}\label{a_cond}
  With reference to Remark~\ref{selec_a}, throughout the reminder of
  this paper (and, more specifically, in connection with the
  accelerated scheme), the parameter $a$ is additionally assumed to
  satisfy the condition $a<L$. Under this assumption, the singular
  integration region (that is, the integration interval
  in~\eqref{dsing_def}) necessarily lies within the union of at most
  three cells $c^q$.
\end{remark}

Taking into account~\eqref{shifted_green_function},
equation~\eqref{D_ns_inf} tells us that the quantity $\tilde
D_{\mathrm{reg}}^{\Delta x}[\mu_1,\dots,\mu_N](x_i)$ equals the field
at the point $x_i$ that arises from free-space ``true'' sources which
are located at points $(x_\ell,f(x_\ell)) - mh e_2$ with
$\ell\in\mathbb{Z}$ and $0\leq m\leq j$, whose $x$-coordinates differ
from $x_i$ in no less than $\gamma$ (see Figure~\ref{POU} and
equation~\eqref{S_def}). Figure~\ref{sources_and_copies}, which
depicts such an array of true sources, displays as black dots
(respectively gray dots) the ``surface true sources''
$(x_\ell,f(x_\ell))$ (resp. the ``shifted true sources''
$(x_\ell,f(x_\ell)) - mh e_2$ with $1\leq m\leq j$).

\begin{figure}[ht!] 
 \centering
 \includegraphics[scale=0.5]{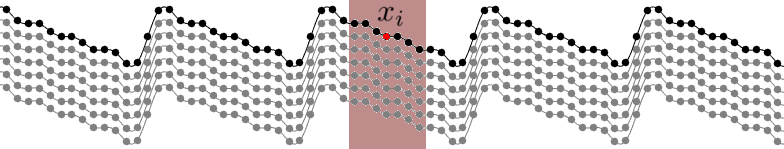}

 \caption{Surface true sources (black), and shifted true sources
   (gray). Matching the color code in Figure~\ref{POU}, the sources
   giving rise to ``local'' interactions for the given target point
   $x_i$ are contained in the region shaded in pink. The accelerated
   algorithm in Section~\ref{conv_ref} below produces $\tilde
   D_\mathrm{reg}^{\Delta x}(x_i)$ (equation~\eqref{D_ns_inf}) by
   subtraction of incorrect local contributions in an FFT-based
   ``all-to-all'' operator, followed by addition of the correct local
   contributions. }
\label{sources_and_copies}
\end{figure}

In order to accelerate the evaluation of the operator $\tilde
D_{\mathrm{reg}}^{\Delta x}$, at first we disregard the shifted true
sources (gray points in Figure~\ref{sources_and_copies}) and we
restrict attention to the surface true-sources (black dots) that are
contained within a given cell $c^q$. In preparation for FFT
acceleration we seek to represent the field generated by the latter
sources in two different ways. As indicated in what follows, the
equivalent sources are to be located in ``Horizontal'' and
``Vertical'' sets $\Lambda_q^H$ and $\Lambda_q^V$ of equispaced
discretization points,
\begin{equation}\label{eq_mesh}
  \Lambda_q^\lambda = \{ \mathbf{y}_s^{\lambda,q} : s=1,\dots,n_\mathrm{eq} \} \qquad (\lambda=H,V),
\end{equation}
contained on (slight extensions of) the horizontal and vertical sides
of $c^q$, respectively; see Figure~\ref{fig_parallel_faces}.  (In the
examples considered in this chapter each one of the extended sets
$\Lambda_q^H$ and $\Lambda_q^V$ contain approximately 20\% more
equivalent-source points than are contained on each pair of parallel
sides of the squares $c^q$ themselves. Such extensions provide slight
accuracy enhancements as discussed in~\cite{BrunoKunyansky}.)  The
resulting equivalent-source approximation, which is described in
detail in Section~\ref{eq_source_I}, is valid and highly accurate
outside the square domain $\mathcal{S}_q$ of side $3L$ and concentric
with $c^q$:
\begin{equation}\label{Sq_notation}
\mathcal{S}_q = \bigcup_{-1 \le m,n \le 1} \Big( c^q + (n,m)L \Big).
\end{equation}

\begin{figure}[ht!]
 \centering
 \includegraphics[width=0.2\textwidth, height=0.2\textwidth]{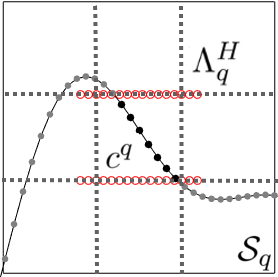} \qquad
 \includegraphics[width=0.2\textwidth, height=0.2\textwidth]{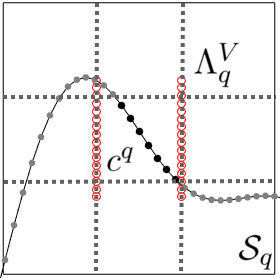} \qquad
 \includegraphics[width=0.2\textwidth, height=0.2\textwidth]{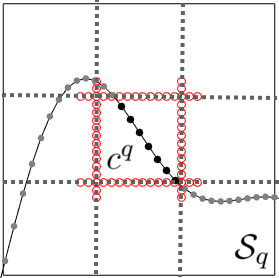}
 \caption{``Free-Space'' Equivalent source geometry. The true sources within $c^q$ (resp. outside $c^q$) are displayed as solid black (resp. gray) circles. The left, center and right images depict in red unfilled circles the horizontal set $\Lambda_q^H$, the vertical set $\Lambda_q^V$, and the union $\Lambda_q^H\cup \Lambda_q^V$, respectively. }
 \label{fig_parallel_faces}
 \end{figure}
\noindent
Importantly, for each $\lambda$ (either $\lambda=H$ or $\lambda=V$) the union of $\Lambda^{\lambda}_q$ for all $q$ (equation~\eqref{grids_def} below) is a Cartesian grid, and thus facilitates evaluation of certain necessary discrete convolutions by means of FFTs, as desired. 

The proposed acceleration procedure is described below, starting with
the computation of the equivalent-source densities
(Section~\ref{eq_source_I}) and following with the incorporation of
shifted equivalent sources and consideration of an associated validity
criterion (Section~\ref{eq_source_II}). This validity criterion
induces a decomposition of the operator $\tilde
D_{\mathrm{reg}}^{\Delta x}$ into two terms
(Section~\ref{reconstruct}), each one of can be produced via certain
FFT-based convolutions (Section~\ref{conv_ref}). A reconstruction of
needed surface fields is then produced (Section~\ref{planewave_rec}),
and, finally, the overall fast high-order solver for
equation~\eqref{eq_per_doublelayer} is presented
(Section~\ref{overall_fast_solv}). For convenience, shifted and
unshifted ``punctured Green functions'' $\Phi_j : \mathbb{R}^2 \times
\mathbb{R}^2 \to \mathbb{C}$ and $\Phi : \mathbb{R}^2 \times
\mathbb{R}^2 \to \mathbb{C}$ are used in what follows which, in terms
of the two-dimensional observation and integration variables
$\mathbf{x}=(x_1,x_2)\in \mathbb{R}^2$ and $\mathbf{y}=(y_1,y_2)\in
\mathbb{R}^2$, are given by
\begin{equation}\label{phi_def}
\Phi_j(\mathbf{x},\mathbf{y}) = \left\lbrace
\begin{array}{cc}
G_j(x_1-y_1, x_2-y_2) & \mbox{ for } \mathbf{x} \ne \mathbf{y} \\
0 & \mbox{ for } \mathbf{x}=\mathbf{y}
\end{array} 
\right.
\quad \mbox{and}\quad \Phi = \Phi_0.
\end{equation}

\subsection{Equivalent-source representation I: surface true sources\label{eq_source_I}}
As indicated above, this section provides an equivalent-source
representation of the contributions to the quantity $\tilde
D_{\mathrm{reg}}^{\Delta x}[\mu_1,\dots,\mu_N](x_i)$
in~\eqref{D_ns_inf} that arise from surface true sources only (the
solid black points in Figure~\ref{sources_and_copies}). To do this we
define
\begin{equation}\label{true_field}
  \psi^{q}(\mathbf{x}) = \sum_{ (x_\ell,f(x_\ell)) \,\in\, c^q } \left( \mu_\ell \; \frac{\partial}{\partial n_y} \Phi(\mathbf{x},\mathbf{y}) \big|_{\mathbf{y} = (x_\ell,f(x_\ell))}\right) (\Delta s)_\ell ,
\end{equation}
which denotes the field generated by all of the surface true-sources located within the cell $c^q$. In the equivalent-source approach, the function $\psi^{q}$ is evaluated, with prescribed accuracy, by a fast procedure based on use of certain ``horizontal'' and ``vertical'' representations, which are valid, within the given accuracy tolerance, for values of $\mathbf{x}$ outside $\mathcal{S}_q$. Each of those representations is given by a sum of monopole and dipole equivalent-sources supported on the corresponding equispaced mesh $\Lambda_q^\lambda$~\eqref{eq_mesh} ($\lambda=H$ or $\lambda=V$). 

To obtain the desired representation a least-squares
problem is solved for each cell $c^q$ (cf.~\cite{BrunoKunyansky}). In detail, for $\lambda=H$ and $\lambda=V$ and for each $q$, an approximate representation of the form
\begin{equation}\label{eq_source_point_x}
  \psi^q(\mathbf{x}) \approx \varphi^{q,\lambda}(\mathbf{x}), \quad\mbox{where}\quad \varphi^{q,\lambda}(\mathbf{x}) = \sum_{s=1}^{n_{eq}} \left(\Phi(\mathbf{x},\mathbf{y}_s^{q,\lambda}) \xi_s^{q,\lambda} + \frac{\partial}{\partial \nu(y)} \Phi(\mathbf{x},\mathbf{y}_s^{q,\lambda}) \zeta_s^{q,\lambda}\right)
\end{equation}
is sought, where $\xi_s^{q,\lambda}$ and $\zeta_s^{q,\lambda}$ are complex numbers (the ``equivalent-source densities''), and where $\nu(y)$ denotes the normal to $\Lambda_q^\lambda$. The densities $\xi_s^{q,\lambda}$ and $\zeta_s^{q,\lambda}$ are obtained as the QR-based solutions~\cite{GolubVanLoan} of the oversampled least-squares problem 
\begin{equation} \label{eq_source_point}
 \min_{(\xi_s^{q,\lambda},\zeta_s^{q,\lambda})} \sum_{t=1}^{n_\mathrm{coll}} \left| \psi^q(\mathbf{x}^q_t) - \sum_{s=1}^{n_{eq}} \left(\Phi(\mathbf{x}^q_t,\mathbf{y}_s^{q,\lambda}) \xi_s^{q,\lambda} + \frac{\partial}{\partial n_y} \Phi(\mathbf{x}^q_t,\mathbf{y}_s^{q,\lambda}) \zeta_s^{q,\lambda} \right)\right|^2,
\end{equation}
where $\{ \mathbf{x}_t^q\}_{t=1,\dots,n_\mathrm{coll}}$ is a sufficiently fine discretization of $\partial\mathcal{S}^{q}$, which in general may be selected arbitrarily, but which we generally take to equal the union of  equispaced discretizations of the sides of $\partial\mathcal{S}^{q}$ (as displayed in Figure~\ref{eqs_validity}). Under these conditions, the equivalent source representation $\varphi^{q,\lambda}$ matches the field values $\psi^q(\mathbf{x})$ for $\mathbf{x}$ on the boundary of $\mathcal{S}^q$ within the prescribed tolerance. Since $\varphi^{q,\lambda}$ and $\psi^q(\mathbf{x})$ are both solutions of the Helmholtz equation with wavenumber $\mathrm k$ outside $\mathcal{S}^q$, it follows that $\varphi^{q,\lambda}$ agrees closely with $\psi^q(\mathbf{x})$ through the exterior of $\mathcal{S}^q$ as well~\cite{BrunoKunyansky}. The equivalent-source approximation and its accuracy outside of $\mathcal{S}_{q}$ is
demonstrated in Figure~\ref{eqs_validity} for the case $\lambda=H$
(``horizontal'' representation).

\begin{figure}[ht!]
 \centering
  \includegraphics[width=0.25\textwidth, height=0.25\textwidth]{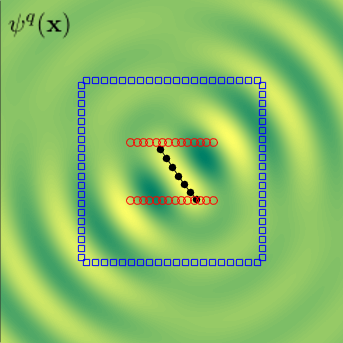} \;
  \includegraphics[width=0.25\textwidth, height=0.25\textwidth]{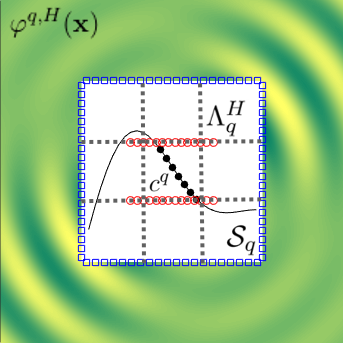} \;
  \includegraphics[width=0.3\textwidth, height=0.25\textwidth]{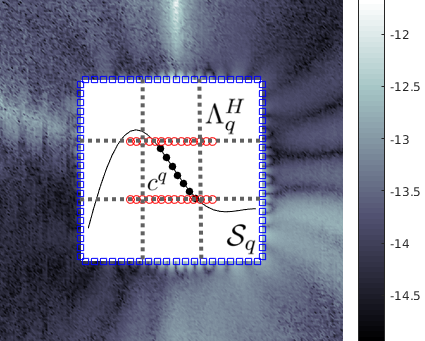}
  
  \caption{Left: Field $\psi^{q}(\mathbf{x})$ generated by the surface
    true sources (solid black circles), evaluated throughout
    space. Center: Approximate field $\varphi^{q,H}(\mathbf{x})$
    generated by the equivalent sources (unfilled red circles),
    evaluated outside $\mathcal{S}^q$.  Right: approximation error
    $|\psi^{q}(\mathbf{x})- \varphi^{q,H}(\mathbf{x})|$ outside
    $\mathcal{S}^q$ (in $\log_{10}$ scale). Collocation points
    $\mathbf{x}_t\in\partial\mathcal{S}^q$ are displayed as blue
    squares. According to the right image, the error for this test
    case ($\mathrm{k}=10$, $L=0.3$, $n_\mathrm{eq}=10$,
    $n_\mathrm{coll}=80$) is smaller than $10^{-12}$ everywhere in the
    validity region
    $\left\{x\in\mathbb{R}^2:x\not\in\mathcal{S}^q\right\}$..}
  \label{eqs_validity}
\end{figure}

\begin{remark}\label{eq_source_quasiper}
  It is easy to see that the equivalent source densities
  $(\xi_s^{q,\lambda},\zeta_s^{q,\lambda})$ are $\alpha$-quasi
  periodic quantities, in the sense that given two cells, $c^{q}$ and
  $c^{q'}$, where $c^{q'}$ is displaced from $c^{q}$, in the
  horizontal direction, by an integer multiple $pd$ of the period $d$,
  we have $(\xi_s^{\lambda,q'},\zeta_s^{\lambda,q'}) =
  e^{\mathrm{i}\alpha pd}(\xi_s^{\lambda,q},\zeta_s^{\lambda,q})$. To
  check this note that, since the corresponding density $\mu_\ell$ is
  itself $\alpha$-quasi periodic (equation~\eqref{quasiper_disc}), in
  view of~\eqref{true_field} it follows that so is the quantity
  $\psi^q(\mathbf{x}^r_t)$ in~\eqref{eq_source_point}. In particular, we have $\psi^{q'}(\mathbf{x}^{q'}_t)= e^{\mathrm{i}\alpha
    pd}\psi^q(\mathbf{x}^q_t)$. Since, additionally,
  $\Phi(\mathbf{x}^{q'}_t,\mathbf{y}_s^{\lambda,q'})=\Phi(\mathbf{x}^q_t,\mathbf{y}_s^{\lambda,q})$,
  we conclude that the least square problems~\eqref{eq_source_point}
  for $q$ and $q'$ are equivalent, and the desired
  $\alpha$-quasiperiodicity of
  $(\xi_s^{q,\lambda},\zeta_s^{q,\lambda})$ follows.
\end{remark}

\subsection{Equivalent-source representation II: shifted true sources\label{eq_source_II}}

In order to incorporate shifted true sources within the equivalent source representation we define the quantity
\begin{equation}\label{true_field_jshifts}
  \psi_j^{q}(\mathbf{x}) = \sum_{ (x_\ell,f(x_\ell)) \,\in\, c^q } \left( \mu_\ell \; \frac{\partial}{\partial n_y} \Phi_j(\mathbf{x},\mathbf{y}) \big|_{\mathbf{y} = (x_\ell,f(x_\ell))} \right) (\Delta s)_\ell
\end{equation}
which, in view of in view of~\eqref{phi_def}, contains some of the
contributions on the right hand side of~\eqref{D_ns_inf}. (With
reference to \eqref{K_def}, note that a term in the sum
\eqref{true_field_jshifts} coincides with a corresponding term
in~\eqref{D_ns_inf} if and only if
$1-S_{\gamma,a}^f(x_i,x_\ell)=1$. For $\mathbf{y} = (x_\ell,f(x_\ell))
\in c^q$, the latter relation certainly holds provided
$\mathbf{x}=(x_i,f(x_i))$ is sufficiently far from $c^q$. But there
are other pairs $(\mathbf{x,y})$ for which this this relation holds;
see Section~\ref{reconstruct} below for details.)

In view of~\eqref{shifted_green_function}, the field $\psi_j^{q}$
in~\eqref{true_field_jshifts} includes contributions from all surface
sources contained within the cell $c^q$ (solid black dots in
Figure~\ref{eqs_validity_shifts}), as well as all of the shifted true
sources that lie below them (which are displayed as gray dots in
Figure~\ref{eqs_validity_shifts}). Importantly, as illustrated in
Figure~\ref{eqs_validity_shifts}, these shifted sources may or may not
lie within $c^q$.

In order to obtain an equivalent-source approximation of the
shifted-true-source quantity $\psi_j^{q}$
in~\eqref{true_field_jshifts} which is analogous to the approximation
\eqref{eq_source_point_x} for the surface true sources, we consider
the easily-checked relation
\begin{equation}\label{psi_jr}
  \psi_j^q(\mathbf{x}) = \sum_{m=0}^{j} (-1)^m C_m^j \psi^{q}(\mathbf{x}-m\bar h),
\end{equation}
where $\bar h= (0,h)$, and we use the
approximation $\psi^q(\mathbf{x}-m\bar h) \approx
\varphi^{q,\lambda}(\mathbf{x}-m\bar h)$ which follows  by
employing~\eqref{eq_source_point_x} at the point $\mathbf{x}-m\bar
h$, for each $m$. Since, in view of the relation
$\Phi(\mathbf{x}+\mathbf{z},\mathbf{y}) =
\Phi(\mathbf{x},\mathbf{y}-\mathbf{z})$, we have
\begin{equation}\label{approx_field_shifted}
\varphi^{q,\lambda}(\mathbf{x}-m\bar h) = \sum_{s=1}^{n_{eq}} \left( \Phi(\mathbf{x},\mathbf{y}_s^{q,\lambda}+m\bar{h}) \xi_s^{q,\lambda} + \frac{\partial}{\partial \nu(y)} \Phi(\mathbf{x},\mathbf{y}_s^{q,\lambda}+m\bar{h}) \zeta_s^{q,\lambda} \right),
\end{equation}
summing~\eqref{approx_field_shifted} over $m$ yields the desired
approximation:
\begin{equation}\label{psi_jr_approx}
  \psi_j^q(\mathbf{x}) \approx   \varphi_j^{q,\lambda}(\mathbf{x}),\quad\mbox{where}\quad  \varphi_j^{q,\lambda}(\mathbf{x}) = \sum_{s=1}^{n_{eq}} \left( \Phi_j(\mathbf{x},\mathbf{y}_s^{q,\lambda}) \xi_s^{q,\lambda} + \frac{\partial}{\partial \nu(y)} \Phi_j(\mathbf{x},\mathbf{y}_s^{q,\lambda}) \zeta_s^{q,\lambda}\right) .
\end{equation}

The shifted-equivalent-source approximation~\eqref{psi_jr_approx} is a
central element of the proposed acceleration approach. Noting that,
for each $m$, the approximation~\eqref{approx_field_shifted} is valid
for points $\mathbf{x}$ outside a translated domain $\mathcal{S}^q -
m\bar h$, it follows that, calling 
\begin{equation}\label{S_jq_hat}
  \widehat{\mathcal{S}}_j^q = \bigcup_{m=0}^{j} \left( \mathcal{S}^q - m \bar{h} \right),
\end{equation}
the overall approximation~\eqref{psi_jr_approx} is valid
for all $\mathbf{x}\not\in\widehat{\mathcal{S}}_j^r$. Thus, letting
\begin{equation}\label{S_jq}
  \mathcal{S}_j^q = \bigcup_{\{ r \, : \, c^r \cap  \widehat{\mathcal{S}}_j^q \ne \emptyset \}} c^r,
\end{equation}
(which equals the smallest union of cells $c^r$ that contains
$\widehat{\mathcal{S}}_j^q$), it follows, in particular,
that~\eqref{psi_jr_approx} is a valid approximation for all
$\mathbf{x}\not\in \mathcal{S}_j^q$.

\begin{figure}[ht!]
 \centering
  \includegraphics[width=0.2\textwidth, height=0.35\textwidth]{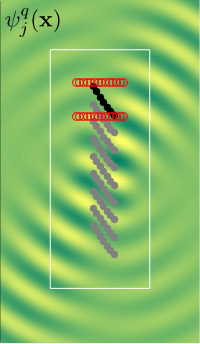} \;
  \includegraphics[width=0.2\textwidth, height=0.35\textwidth]{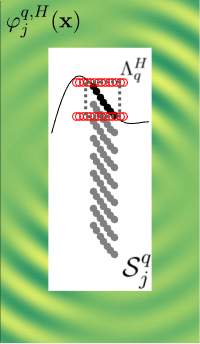} \;
  \includegraphics[width=0.28\textwidth, height=0.35\textwidth]{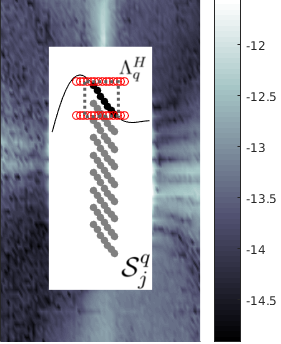}
  
  \caption{Left: Field $\psi_j^q(\mathbf{x})$ generated by the
    combination of the true sources (solid black circles) and their
    shifted copies (solid gray circles).  Center: Approximate field
    $\varphi_j^{q,H}(\mathbf{x})$ generated by shifted equivalent
    sources (unfilled red circles).  Right: approximation error
    $|\psi_j^{q}(\mathbf{x})- \varphi_j^{q,H}(\mathbf{x})|$ outside
    $\mathcal{S}_j^q$ (in $\log_{10}$ scale).  According to the right
    image, the error for this test case ($\mathrm{k}=10$, $L=0.3$,
    $n_\mathrm{eq}=10$, $n_\mathrm{coll}=80$) is smaller than
    $10^{-12}$ everywhere in the validity region
    $\left\{x\in\mathbb{R}^2:x\not\in\mathcal{S}_j^q\right\}$.}
  \label{eqs_validity_shifts}
\end{figure}

\subsection{Decomposition of $\tilde D_{\mathrm{reg}}^{\Delta x}$ in
  ``intersecting'' and ``non-intersecting''
  contributions\label{reconstruct}}

This section introduces a decomposition of the operator $\tilde
D_{\mathrm{reg}}^{\Delta x}$ as a sum of two terms. With reference to
Remark~\ref{cell_def}, and denoting by $\overline A$ the closure of a
set $A$ (the union of the set and its boundary), we define the first
term as
\begin{equation}\label{shifted_true_source}
  \psi_j^{ni,q}:\overline{c^q}\to\mathbb{C}, \qquad \psi^{ni,q}_j(\mathbf{x}) = \sum_{\{r\in\mathbb{Z} \, : \, c^q \cap \mathcal{S}_j^r = \emptyset \}}  \psi^{r}_j(\mathbf{x}),\qquad (\mathbf{x}\in \overline{c^q}).
\end{equation}
Clearly, for $\mathbf{x}\in \overline{c^q}$, the quantity
$\psi_j^{ni,q}(\mathbf{x})$ contains the ``non-intersecting''
contributions---that is, contributions arising from sources contained
in cells $c^r \subset \Omega_\infty$, $-\infty\leq r\leq\infty$
(cf. equation~\eqref{Omegas}), such that $\mathcal{S}_j^r$ does not
intersect $c^q$, and for which therefore, according to
Section~\ref{eq_source_II}, the approximation~\eqref{psi_jr_approx} is
valid. Since the operator $\tilde D_{\mathrm{reg}}^{\Delta x}$ acts on
spaces of functions defined on $\Gamma$, in what follows we will often
evaluate $\psi_j^{ni,q}$ at points $\mathbf{x}\in \overline{c^q}$ of
the form $\mathbf{x} = (x_i,f(x_i))\in \overline{c^q}$.  For
$\mathbf{x} = (x_i,f(x_i))\in \overline{c^q}$, then, the second term
equals, naturally,
\begin{equation}\label{final_correction}
 \psi^{\textit{int},q}_j\Big( x_i,f(x_i) \Big) =  \tilde D_{\mathrm{reg}}^{\Delta x}[\mu_1,\dots,\mu_N](x_i) - \psi^{ni,q}_j\Big( x_i,f(x_i) \Big),\qquad (x_i,f(x_i))\in
c^q.
\end{equation}

In view of definitions~\eqref{true_field_jshifts}
and~\eqref{shifted_true_source}, the field $\psi_j^{ni,q}(\mathbf{x})$
can alternatively by expressed in the form
\begin{equation}\label{shifted_true_source_sum}
  \psi_j^{ni,q}(\mathbf{x}) = \sum_{ \{ \ell \in \mathbb{Z} \, :\, (x_\ell,f(x_\ell)) \,\in\, c^r \,\footnotesize{\mbox{with}}\; c^q \cap \mathcal{S}_j^r = \emptyset \}  } \left( \mu_\ell \; \frac{\partial}{\partial n_y} \Phi_j(\mathbf{x},\mathbf{y}) \big|_{\mathbf{y} = (x_\ell,f(x_\ell))} \right) (\Delta s)_\ell.
\end{equation} 
Note that all non-intersecting contributions to a point $\mathbf{x} =
(x_i,f(x_i)) \in c^q$ arise from integration points
$(x_\ell,f(x_\ell))$ that are at a distance larger than the side $L$
of $c^q$. Thus, in view of Remark~\ref{a_cond} ($a < L$), we have
$S_{\gamma,a}^f(x_i,x_\ell)=0$ for all non-intersecting
contributions. In view of~\eqref{D_ns_inf}
and~\eqref{shifted_true_source_sum}, then, we obtain
\begin{equation}\label{final_correction_0}
  \psi^{\textit{int},q}_j\Big( x_i,f(x_i) \Big) =  \sum_{ \{ \ell \in \mathbb{Z} \, :\, (x_\ell,f(x_\ell)) \,\in\, c^r \,\small{\mbox{with}}\; c^q \cap \mathcal{S}_j^r \ne \emptyset \}} \mu_\ell (1-S_{\gamma,a}^f(x_i,x_\ell)) \partial_{\nu'} G_j(x_i,x_\ell) (\Delta s)_\ell.
\end{equation}

The following two sections (\ref{conv_ref} and~\ref{planewave_rec})
present an efficient evaluation strategy for the quantities
$\psi^{ni,q}_j(\mathbf{x})$ over $\Gamma\cap c^q$, for
$q=1,\dots,n_\mathrm{cell}$. This strategy relies on use of the approximation
\begin{equation}\label{phi_na_j_conv}
\psi^{ni,q}_j(\mathbf{x}) \approx \varphi^{ni,q,\lambda}_j(\mathbf{x}); \qquad \varphi^{ni,q,\lambda}_j(\mathbf{x}) = \sum_{ \{r\in\mathbb{Z} \, : \, c^q \cap \mathcal{S}_j^r = \emptyset \} }  \; \sum_{s=1}^{n_\mathrm{eq}} \left( \Phi_j(\mathbf{x},\mathbf{y}_s^{r,\lambda}) \xi_s^{r,\lambda} + \frac{\partial}{\partial n_y} \Phi_j(\mathbf{x},\mathbf{y}_s^{r,\lambda}) \zeta_s^{r,\lambda}\right),
\end{equation}
(for $\mathbf{x}\in c^q$) which can be obtained by substituting
equation~\eqref{psi_jr_approx} into~\eqref{shifted_true_source}. As
shown in Section~\ref{conv_ref}, the quantities $
\varphi^{ni,q,\lambda}_j(\mathbf{x})$ in~\eqref{phi_na_j_conv}
($q=1,\dots,n_\mathrm{cell}$) are related to a single discrete
Cartesian convolution that can be evaluated rapidly by means of the
FFT algorithm. Once $\psi^{ni,q}(\mathbf{x})$ has been evaluated (by
means of $ \varphi^{ni,q,\lambda}_j$), the remaining ``local''
contributions $\psi^{\textit{int},q}_j$ to $\tilde
D_{\mathrm{reg}}^{\Delta x}$ can be incorporated
using~\eqref{final_correction_0} at a small computational cost. The
overall fast high-order numerical algorithm for evaluation of the
operator on the left-hand side of equation~\eqref{eq_per_doublelayer}
(which also incorporates the implementations of the operators $\tilde
D^{\Delta x}_{\mathrm{sing}}$ and $D_M^{\Delta x}$ presented in
Section~\ref{Overall}) together with the associated fast iterative
solver, are then summarized in Section~\ref{overall_fast_solv}.

\subsection{Approximation of $\psi^{ni,q}_j$ via global and local convolutions at FFT speeds\label{conv_ref}}

In order to accelerate the evaluation of $\psi^{ni,q}_j$ by means of the
FFT algorithm we introduce the quantity
\begin{equation}\label{large_convolution}
 \varphi^{\textit{all},\lambda}_j(\mathbf{x}) = \sum_{r \in \mathbb{Z}}  \sum_{s=1}^{n_\mathrm{eq}} \left( \Phi_j(\mathbf{x},\mathbf{y}_s^{r,\lambda}) \xi_s^{r,\lambda} + \frac{\partial}{\partial n_y} \Phi_j(\mathbf{x},\mathbf{y}_s^{r,\lambda}) \zeta_s^{r,\lambda} \right)
\end{equation}
which incorporates the non-intersecting terms already included
in~\eqref{phi_na_j_conv} as well as undesired ``intersecting'' (local)
terms. For each $q$, the sum of all undesired intersecting terms for
the domain $\overline{c^q}$ is a function
$\varphi^{\textit{int},q,\lambda}_j: \overline{c^q} \to \mathbb{C}$
given by
\begin{equation}\label{small_convolution}
 \varphi^{\textit{int},q,\lambda}_j(\mathbf{x}) = \sum_{ \substack{r\in \mathcal{L}(q)\\ 1\le s \le n_\mathrm{eq}} } \left( \Phi_j(\mathbf{x},\mathbf{y}_s^{r,\lambda}) \xi_s^{r,\lambda} + \frac{\partial}{\partial n_y} \Phi_j(\mathbf{x},\mathbf{y}_s^{r,\lambda}) \zeta_s^{r,\lambda}\right), \quad\mbox{where}\; \mathcal{L}(q) = \{ r\in\mathbb{Z} : c^q \subset \mathcal{S}_j^r\}.
\end{equation}
Since, by construction, $c^q \cap \mathcal{S}^r_j\ne \emptyset$ if and
only if $c^q \subset \mathcal{S}^r_j$, in view of~\eqref{phi_na_j_conv}
we clearly have
\begin{equation}\label{psi_na_sum}
   \varphi^{ni,q,\lambda}_j = \varphi^{\textit{all},\lambda}_j - \varphi^{\textit{int},q,\lambda}_j.
\end{equation}
This relation reduces the evaluation of $\varphi^{ni,q,\lambda}_j$ to
evaluation of the $q$-independent quantity~\eqref{large_convolution}
and the $q$-dependent quantity~\eqref{small_convolution}.

The expression~\eqref{large_convolution} for
$\varphi^{\textit{all},\lambda}_j$ requires the evaluation of an
infinite sum. Exploiting the fact that, as indicated in
Remark~\ref{eq_source_quasiper}, the equivalent sources
$(\xi_s^{r,\lambda},\zeta_s^{r,\lambda})$ are $\alpha$-quasi-periodic
quantities, a more convenient expression can be obtained. Indeed,
defining
\begin{equation}\label{phi_jq_def}
\widetilde \Phi_j^\textit{qper} : \mathbb{R}^2 \times
\mathbb{R}^2 \to \mathbb{C},\qquad \widetilde \Phi_j^\textit{qper}(\mathbf{x},\mathbf{y}) = \left\lbrace
\begin{array}{cc}
\tilde G_j^\textit{qper}(x_1-y_1, x_2-y_2) & \mbox{ for } \mathbf{x} \ne \mathbf{y} \\
0 & \mbox{ for } \mathbf{x}=\mathbf{y}
\end{array} 
\right.
\end{equation}
in terms of the variables $\mathbf{x}=(x_1,x_2)\in \mathbb{R}^2$ and
$\mathbf{y}=(y_1,y_2)\in \mathbb{R}^2$, we can express
$\varphi^{\textit{all},\lambda}_j$ as the sum
\begin{equation}\label{large_convolution_per}
 \varphi^{\textit{all},\lambda}_j(\mathbf{x}) = \sum_{r=1}^{n_\mathrm{cell}} \sum_{s=1}^{n_\mathrm{eq}} \left( \widetilde \Phi_j^{\textit{qper}}(\mathbf{x},\mathbf{y}_s^{r,\lambda}) \xi_s^{r,\lambda} + \frac{\partial}{\partial n_y} \widetilde \Phi_j^{\textit{qper}}(\mathbf{x},\mathbf{y}_s^{r,\lambda}) \zeta_s^{r,\lambda}\right)
\end{equation}
of finitely many terms, each one of which contains $\widetilde \Phi_j^{\textit{qper}}$.

In order to evaluate the quantities $\varphi^{\textit{all},\lambda}_j$
and $\varphi^{\textit{int},q,\lambda}_j$ by means of the FFT algorithm
we use the equivalent-source meshes $\Lambda_q^\lambda$ introduced in
Section~\ref{sec_shifted_eqs} (and depicted in
Figure~\ref{fig_parallel_faces}) and we define, for $\lambda=H,V$, the
``global'' and ``local'' Cartesian grids
\begin{equation}\label{grids_def}
 \Pi_\lambda^{\textit{per}} = \bigcup_{ \{ r\in\mathbb{Z} \; : \; c^r \subseteq \Omega_{per} \}} \Lambda^\lambda_r \qquad \mbox{and} \qquad \Pi_\lambda^q =  \bigcup_{ \{ r\in\mathbb{Z} \; : \; c^q \subset \mathcal{S}_j^r \}} \Lambda^\lambda_r.
\end{equation}
The following two sections describe algorithms which rapidly evaluate
these quantities by means of FFTs. The evaluation of $\psi^{ni,q}_j$
(which is the main goal of Section~\ref{conv_ref}) then follows
directly, as indicated in Section~\ref{sec_bnd_cq}.

\subsubsection{Evaluation of $\varphi^{\textit{all},\lambda}_j$ in $\Pi_\lambda$ via a global convolution \label{sec_large_conv}}

In order to express $\varphi^{\textit{all},\lambda}_j$ as a
convolution, for $\mathbf{y'}\in\Pi_\lambda^{\textit{per}}$ and
$\lambda = H, V$ we define the sums
\begin{equation}\label{sum_eqsources}
 \xi^{\textit{all},\lambda}(\mathbf{y'}) =  \sum_{\substack{1\le r \le n_\mathrm{cell} \\  1\le s \le n_\mathrm{eq} \\ \mathbf{y}_s^{r,\lambda} = \mathbf{y'} } } \xi_s^{r,\lambda}
 \qquad \mbox{ and } \qquad \zeta^{\textit{all},\lambda}(\mathbf{y'}) = \sum_{\substack{1\le r \le n_\mathrm{cell} \\  1\le s \le n_\mathrm{eq} \\ \mathbf{y}_s^{r,\lambda} = \mathbf{y'}}} \zeta_s^{r,\lambda}
\end{equation}
of equivalent source densities $\xi_s^{r,\lambda}$ and
$\zeta_s^{r,\lambda}$, respectively ($1\le r \le n_\mathrm{cell}$),
that are supported at a given point $\mathbf{y'}\in
\Pi_\lambda^{\textit{per}}$. We note that two and even four
contributions may arise at a point $\mathbf{y'}\in
\Pi_\lambda^{\textit{per}}$---as $\mathbf{y'}$ may lie on a common
side of two neighboring cells, and, in some cases, on the intersection
of four different sets $\Lambda_q^\lambda$---on account of overlap of
the extended regions described in Section~\ref{sec_shifted_eqs} and
depicted in Figure~\ref{fig_parallel_faces}.

Replacing~\eqref{sum_eqsources} in~\eqref{large_convolution_per}, we
arrive at the discrete-convolution expression
\begin{equation}\label{large_convolution_per_grid}
 \varphi^{\textit{all},\lambda}_j(\mathbf{x}) = \sum_{\mathbf{y'} \, \in \Pi_\lambda^{\textit{per}}} \left( \widetilde \Phi_j^{\textit{qper}}(\mathbf{x},\mathbf{y'}) \xi^{\textit{all},\lambda}(\mathbf{y'}) + \frac{\partial}{\partial n_y} \widetilde \Phi_j^{\textit{qper}}(\mathbf{x},\mathbf{y'}) \zeta^{\textit{all},\lambda}(\mathbf{y'})\right),\quad \mathbf{x}\in \Pi_\lambda^{\textit{per}},
\end{equation}
for the quantity $\varphi^{\textit{all},\lambda}_j$ on the mesh
$\Pi_\lambda^{\textit{per}}$. The evaluation of this convolution can
be performed by a standard FFT-based procedure in $O(M \log M)$
operations, where $M = O(n_{\mathrm{cell}}n_{eq})$ denotes the number
of elements in $\Pi_\lambda^{\textit{per}}$. Note that, per
equation~\eqref{phi_jq_def}, this global FFT algorithm requires the
values of the quasi-periodic Green function $\tilde
G^{\textit{qper}}_j(X,Y)$ (see also Remark~\ref{large_FFT_eval} below)
at points $(X,Y)$ in the ``evaluation grid''
$\widehat{\Pi}_\lambda^{\textit{per}} = \lbrace \mathbf{x}-\mathbf{y}
: \mathbf{x},\mathbf{y} \in \Pi_\lambda^{\textit{per}} \rbrace$. In
fact, this is the only point in the accelerated algorithm that
requires use of the quasi-periodic Green function.

\begin{remark}\label{large_FFT_eval}
  An efficient strategy for the evaluation of $\tilde
  G^{\textit{qper}}_j(X,Y)$ at a given point was presented in
  Section~\ref{hybrid_spectral}, which makes use of both spectral and
  spatial representations of this function. Additional performance
  gains are obtained in the present context by exploiting certain
  symmetries in the evaluation grid
  $\widehat{\Pi}_\lambda^{\textit{per}}$. The identity $\tilde
  G^{\textit{qper}}_j(X+d,Y)=e^{i\alpha d}\tilde
  G^{\textit{qper}}_j(X,Y)$ is used to restrict the evaluation of the
  function $\tilde G^{\textit{qper}}_j(X,Y)$ at, say, only positive
  values of $X$; for the $j=0$ case, the identity $\tilde
  G^{\textit{qper}}_0(X,Y)=\tilde G^{\textit{qper}}_0(X,-Y)$ is
  similarly used to restrict evaluation of $\tilde
  G^{\textit{qper}}_0(X,Y)$ to positive values of $Y$. Further, since
  the spectral series~\eqref{shifted_periodic_green_function_spectral}
  is a sum of exponentials which can expressed as products of
  exponentials that depend on $X$ and $Y$ separately, the spectral
  series can be evaluated efficiently by utilizing precomputed values
  of the required single-variable exponentials---with limited
  computing and storage cost. For an efficient implementation of the
  spatial series, finally, asymptotic expansions of the Hankel
  functions as proposed in~\cite{BrunoDelourme} are also used. The
  overall strategy produces the required values of $\tilde
  G^{\textit{qper}}_j$ over the necessary evaluation grid
  $\widehat{\Pi}_\lambda^{\textit{per}}$ in a highly efficient manner.
\end{remark}

\subsubsection{Evaluation of $\varphi^{\textit{int},q,\lambda}_j$ in $\Pi_\lambda^q$ via a local convolution \label{sec_small_conv}}

In order to express $\varphi^{\textit{int},q,\lambda}_j$ (equation~\eqref{small_convolution}) as a convolution, for $\mathbf{y'}\in\Pi_\lambda^q$, we define the sums 
\begin{equation}\label{sum_eqsources_local}
\xi^{q,\lambda}(\mathbf{y'}) = \sum_{ \substack{ r \in \mathcal{L}(q)  \\  1 \le s \le n_\mathrm{eq} \\ \mathbf{y}_s^{r,\lambda} = \, \mathbf{y'} }} \xi_s^{r,\lambda}
 \qquad \mbox{ and } \qquad
 \zeta^{q,\lambda}(\mathbf{y'}) = \sum_{ \substack{ r \in \mathcal{L}(q)  \\  1 \le s \le n_\mathrm{eq} \\ \mathbf{y}_s^{r,\lambda} = \, \mathbf{y'} }} \zeta_s^{r,\lambda}
\end{equation}
of equivalent source densities $\xi_s^{r,\lambda}$ and
$\zeta_s^{r,\lambda}$, respectively, that are supported at the point
$\mathbf{y'}\in \Pi_\lambda^q$, where $r$ lies in the local set of
indexes $\mathcal{L}(q)$ defined in~\eqref{small_convolution}.  Note
that, the set $\mathcal{L}(q)$ contains integers $r$ that may lie
outside the range $1 \le r \le n_\mathrm{cell}$.  In such cases, in
order to avoid use of equivalent source densities that lie outside the
reference periodicity domain $\Omega^\textit{per}$, the
$\alpha$-quasiperiodicity of $\xi_s^{r,\lambda}$ and
$\zeta_s^{r,\lambda}$ (Remark~\ref{eq_source_quasiper}) is utilized to
re-express the sums in~\eqref{sum_eqsources_local} in terms of
equivalent sources $\xi_s^{r,\lambda}$ and $\zeta_s^{r,\lambda}$ for
which $1\le r \le n_\mathrm{cell}$. Additionally note that, as in
Section~\ref{sec_large_conv}, two and even four contributions may
arise in the sums~\eqref{sum_eqsources_local} for a given point
$\mathbf{y'}\in \Pi_\lambda^q$.

Replacing~\eqref{sum_eqsources_local} in~\eqref{small_convolution} yields the discrete-convolution expression
\begin{equation}\label{small_convolution_grid}
 \varphi^{\textit{int},q,\lambda}_j(\mathbf{x}) = \sum_{\mathbf{y'} \, \in \Pi_\lambda^q} \left( \Phi_j(\mathbf{x},\mathbf{y'}) \xi^{q,\lambda}(\mathbf{y'}) + \frac{\partial}{\partial n_y} \Phi_j(\mathbf{x},\mathbf{y'}) \zeta^{q,\lambda}(\mathbf{y'}) \right), \quad \mathbf{x}\in\Pi_\lambda^q,
\end{equation}
which can be evaluated for all $\mathbf{x}\in \Pi_\lambda^q$ by means of an FFT procedure, in $O(M_q \log M_q)$ operations, where $M_q = O(n_{eq})$ denotes the number of elements in $\Pi^q_\lambda$. This time, the Green function $\Phi_j$ has to be evaluated on the ``evaluation grid'' 
\begin{equation}\label{pi_q_hat}
\widehat{\Pi}_\lambda^{q} = \lbrace \mathbf{x}-\mathbf{y}
: \mathbf{x},\mathbf{y} \in \Pi_\lambda^{q} \rbrace.
\end{equation}
Notice that the set $\widehat{\Pi}_\lambda^{q}$ is in fact independent of $q$.
 
\subsubsection{Approximation of $\psi^{ni,q}_j$ on the boundary of $c^q$ \label{sec_bnd_cq}}

Having obtained $\varphi^{\textit{int},q,\lambda}_j(\mathbf{x})$ and
$\varphi^{\textit{all},\lambda}_j(\mathbf{x})$, the desired quantities
$\varphi^{ni,q,\lambda}_j(\mathbf{x})$, for $\lambda=H,V$ follow
from~\eqref{psi_na_sum}. For each $q$ ($1\leq q \le n_\mathrm{cell}$)
the resulting discrete values $\varphi^{ni,q,H}_j$ and
$\varphi^{ni,q,V}_j$ are finally
 used to form the mesh functions
\begin{equation}\label{combine_psi_pna}
\varphi^{ni,q}_j: \overline{c^q}
\cap (\Pi_H^q\cup \Pi_V^q)\to \mathbb{C}, \qquad \varphi^{ni,q}_j(\mathbf{x}) = \varphi^{ni,q,\lambda}_j(\mathbf{x}) \quad \mbox{ for } \quad \mathbf{x}\in \overline{c^q}\cap
\Pi_\lambda^q\quad\mbox{and} \quad \lambda=H,V.
\end{equation}
It is clear, by construction, that $\varphi^{ni,q}_j(\mathbf{x})$ is
an approximation of $\psi^{ni,q}_j(\mathbf{x})$ for each element
$\mathbf{x}$ in the discretization $\overline{c^q} \cap (\Pi_H^q\cup
\Pi_V^q)$ of the boundary of $c^q$. Using these approximate values,
the next section presents a method for the high-order evaluation of
$\psi^{ni,q}_j(\mathbf{x})$ at an arbitrary point within $c^q$, and
thus, in particular, on the portion $\Gamma\cap c^q$ of the scattering
surface $\Gamma$ contained within $c^q$.

\subsection{Plane Wave representation of $\psi^{ni,q}_j$ \label{planewave_rec} within $c^q$}

Since $\psi^{ni,q}_j(\mathbf{x})$ satisfies the Helmholtz equation
within the cell $c^q$, and in view of Remark~\ref{non-resonant}, this
field can be obtained within that cell as the solution of the
Dirichlet problem with values $\psi^{ni,q}_j(\mathbf{x})$ on the cell
boundary. Using the approximate values $\varphi^{ni,q}_j$ of the field
$\psi^{ni,q}_j(\mathbf{x})$ that are produced, on the discrete mesh
$\overline {c^q}\cap(\Pi_H^q\cup\Pi_V^q)$, by the fast algorithm
described in Section~\ref{conv_ref}, approximate values of the
solution $\psi^{ni,q}_j(\mathbf{x})$ of this Dirichlet problem for
$\mathbf{x}\in c^q$ are obtained~\cite{BrunoKunyansky} by means of a
discrete plane wave expansion. Thus, using a number $n_\mathrm{plw}$
of plane waves, the proposed approximation for $\mathbf{x} \in c^q$ is
thus given by the expression
\begin{equation}\label{planewave}
\psi^{ni,q}_j(\mathbf{x}) \approx  \eta^{ni,q}_j(\mathbf{x}) \quad \mbox{where} \quad  \eta^{ni,q}_j(\mathbf{x})= \sum_{s=1}^{n_\mathrm{plw}} { w_{s}.e^{\mathrm{i}\mathrm{k} d_s \cdot \mathbf{x}  }},\quad \mathbf{x}\in c^q,
\end{equation}
where the weights $w_{i}$ are obtained as the QR solution~\cite{GolubVanLoan}
of the least squares problem
\begin{equation}\label{plwave_QR}
  \min_{\{ w_{s} \}} \sum_{\mathbf{x}\in \overline {c^q}\cap(\Pi_H^q\cup\Pi_V^q)} \left| \varphi^{ni,q}_j(\mathbf{x}) - \sum_{s=1}^{n_\mathrm{plw}} { w_{s}.e^{\mathrm{i}\mathrm{k} d_s \cdot \mathbf{x}  }} \right|^2, \quad \mbox{where} \quad d_s = \left( \sin\left(\frac{2 \pi s}{n_\mathrm{plw}}\right), \cos \left(\frac{2 \pi s}{n_\mathrm{plw}}\right) \right).
\end{equation}

This is the last necessary element in the proposed algorithm for fast
approximate evaluation of the operator $\tilde
D_{\mathrm{reg}}^{\Delta x}$. Using the various components introduced
above in the present Section~\ref{accel},
Section~\ref{overall_fast_solv} describes the overall proposed fast
high-order solver. 

\subsection{Overall fast high-order solver for
  equation~\eqref{eq_per_doublelayer}\label{overall_fast_solv}}

The overall solver described in what follows results as a modified
version of the unaccelerated solver presented in
Section~\ref{Overall}: in the present accelerated solver the
evaluation of the operator $\tilde D_{\mathrm{reg}}^{\Delta x}$ is
carried out using the procedure described in
Sections~\ref{sec_shifted_eqs} through~\ref{planewave_rec} instead of
the straightforward $O(N^2)$ approach used in Section~\ref{Overall}.
Algorithms 1 to 3 summarize the overall accelerated solution method.
\begin{algorithm}[H]
  \caption{Main program: solution of equation~\eqref{discrete_eq}}
  \label{alg_all}
\begin{algorithmic}
  \STATE Run Initialization (Algorithm 2) \STATE Run GMRES iterations,
  using the forward-map Algorithm 3, on the linear algebra
  problem~\eqref{discrete_eq}
\end{algorithmic}
\end{algorithm}

\begin{algorithm}
  \caption{Initialization}
  \label{alg_init}
\begin{algorithmic}
  \STATE Obtain QR factors for~\eqref{eq_source_point}
  and~\eqref{plwave_QR} \COMMENT{Only once (they do not depend on
    $q$).}  \STATE Evaluate $G_j^{\textit{qper}}$ on
  $\widehat{\Pi}_\lambda^{\textit{per}}$
  \COMMENT{Remark~\ref{large_FFT_eval}.}  \STATE Evaluate $G_j$ on
  $\widehat{\Pi}_\lambda^{q}$ \COMMENT{Only once
    ($\widehat{\Pi}_\lambda^{q}$ in~\eqref{pi_q_hat} does not depend
    on $q$).}  \STATE Precompute matrices for $\tilde D^{\Delta
    x}_{\mathrm{sing}}$ and $D^{\Delta x}_M$
  \COMMENT{Equations~\eqref{local_mat} and~\eqref{D_M_deltax}.}
\end{algorithmic}
\end{algorithm}

\begin{algorithm}[H]
\caption{Discrete forward map: $[\mu_1,\dots,\mu_n]\to \left(\frac{1}{2}I + D^{\Delta x}\right)[\mu_1,\dots,\mu_n]$ }
\label{alg_iter}
\begin{algorithmic}
  \STATE $\{(\xi_s^{q,\lambda},\zeta_s^{q,\lambda})\} \longleftarrow$
  \texttt{EqSources} \COMMENT{Solve least squares
    problem~\eqref{eq_source_point}.}  \STATE
  $\{(\xi^{\textit{all},\lambda}_{\mathbf{y'}},\zeta^{\textit{all},\lambda}_{\mathbf{y'}})\}
  \longleftarrow$ \texttt{GlobalEqSMerge} \COMMENT{Combine equivalent
    sources~\eqref{sum_eqsources}.}  \STATE
  $\{\varphi^{all,\lambda}_j\} \longleftarrow$ \texttt{GlobalFFT}
  \COMMENT{Evaluate~\eqref{large_convolution_per_grid} via FFT on the grid $\Pi_\lambda^{\textit{per}}$.}  \STATE
  $\{(\xi^{q,\lambda}_{\mathbf{y'}},\zeta^{q,\lambda}_{\mathbf{y'}})\}
  \longleftarrow$ \texttt{LocalEqSMerge} \COMMENT{Combine equivalent
    sources~\eqref{sum_eqsources_local}.}  \STATE
  $\{\varphi^{\textit{int},q,\lambda}_j\} \longleftarrow$
  \texttt{LocalFFT} \COMMENT{Evaluate~\eqref{small_convolution_grid}
    via FFT on the grid $\Pi_\lambda^q$.}  
    \STATE $\{\varphi^{ni,q,\lambda}_j\}
  \longleftarrow$ \texttt{LocalSubtract} \COMMENT{Subtract
    $\varphi^{\textit{int},q,\lambda}_j$ from $\varphi^{all,\lambda}_j$~\eqref{psi_na_sum}.}
  \STATE $\{\varphi^{ni,q}_j\}
  \longleftarrow$ \texttt{Combine-$\lambda$} \COMMENT{Combine $\varphi^{ni,q,H}_j$ and $\varphi^{ni,q,V}_j$ as in~\eqref{combine_psi_pna}.}    
  \STATE $\{w_s^q\} \longleftarrow$ \texttt{PlaneWaveWeights}
  \COMMENT{Solve least square problem~\eqref{plwave_QR}.} \STATE $\psi^{ni,q}_j \dashleftarrow$
  \texttt{NonIntersecting} \COMMENT{Use~\eqref{planewave};
    $\mathbf{x}=(x_i,f(x_i))\in c^q$.}
  \STATE $\psi^{\textit{int},q}_j \longleftarrow$
  \texttt{Intersecting} \COMMENT{Use~\eqref{final_correction_0};
    $(x_i,f(x_i))\in c^q$.} \STATE
  $\tilde D_{\mathrm{reg}}^{\Delta x} \dashleftarrow$
  \texttt{EvalRegular} \COMMENT{Use~\eqref{final_correction},
    $\psi^{ni,q}_j(x_i,f(x_i)) $,
    $\psi^{\textit{int},q}_j(x_i,f(x_i))$.} \STATE $\tilde D_{\mathrm{sing}}^{\Delta x}
  \longleftarrow$ \texttt{EvalSingular} \COMMENT{Use~\eqref{local_mat}; $1\le i \le N$.} \STATE $D^{\Delta
    x}_M \longleftarrow$ \texttt{EvalModes} \COMMENT{Use~\eqref{D_M_deltax};  $1\le i \le N$.}  \STATE $ \left(\frac{1}{2}I + D^{\Delta x}\right)
  \longleftarrow$ \texttt{AddOperators} \COMMENT{Add $\frac{1}{2}I$, $\tilde  D_{\mathrm{sing}}^{\Delta x}$, $\tilde D_{\mathrm{reg}}^{\Delta
      x}$ and $D^{\Delta x}_M$~\eqref{D_deltax}-\eqref{discrete_eq}.}
\end{algorithmic}
\end{algorithm}
\vspace{-0.4cm}
\noindent\small 
Algorithm 3: Routines \texttt{EqSource}, \texttt{GlobalFFT},
etc. perform the tasks described in the corresponding comments on the
right column, resulting on the values indicated by the left-pointing
solid arrows. Dashed arrows indicate that an additional approximation
is used in the assignment.  Whenever the resulting values (on the
left) depend on $q$ and/or $\lambda$, the operations are performed for
$1\leq q\leq n_\mathrm{cell}$ and/or for $\lambda = H,V$,
respectively.  \normalsize \vspace{0.4cm}

The accuracy and efficiency of this algorithm is demonstrated in the
following section.

\begin{remark} Once a solution $\mu$ of the integral
  equation~\eqref{eq_per_doublelayer} has been obtained, a single
  application of a slightly modified version of Algorithm 3 enables
  the evaluation of the scattered field $u^{scat}(\mathbf{x})$
  in~\eqref{eq:int_eq_dirichlet}, and thus the total field
  $u(\mathbf{x}) = u^{scat}(\mathbf{x})+ u^{inc}(\mathbf{x})$, at all
  points $\mathbf{x}=(x,y)$ in a given two-dimensional domain---at a
  very moderate additional computational cost. In brief, the modified
  evaluation procedure only requires that equations
  \eqref{final_correction_0} and~\eqref{planewave}, together with
  their dependencies, be implemented so as to produce the necessary
  scattered field $u^{scat}$ at all points where the fields are desired.
\end{remark}

\section{Numerical results\label{numer}} 

This section presents results of applications of the proposed
algorithm to problems of scattering by perfectly conducting periodic
rough surfaces, at both Wood and non-Wood configurations, with
sinusoidal and composite rough surfaces (including randomly rough
Gaussian surfaces), and through wide ranges of problem
parameters---including grazing incidences and high
period-to-wavelength and/or height-to-period ratios. The presentation
is prefaced by a brief section concerning computational costs. For
brevity, only results for the accelerated method are presented. In all
cases these results compare favorably, in terms of computing times,
accuracy and generality, with those provided by previous
approaches. All computational results presented in this section were
obtained from single-core runs on a 3.4GHz Intel i7-6700 processor
with 4 Gb of memory.

\subsection{Computing costs\label{chap2_comp_cost}} 

The dependence of the computing cost of the algorithm on the size of
the problem is subtle, as it includes costs components from various
code elements (acceleration, integration, Green function evaluations,
etc.), each one of which depends significantly on a variety of
structural parameters---including the shift-parameter $h$, the various
ratios $H/d$, $H/\lambda$, $d/\lambda$ involving the height $H$, the
period $d$, and the wavelength $\lambda$, and the ``roughness'' of the
surface, as quantified by the decay of the associated
spectrum. Roughly speaking, however, the results in the present
section suggest two important asymptotic regimes exist:
(1)~$d/\lambda$ grows as $H/\lambda$ is kept fixed; and, (2)~Both
$d/\lambda$ and $H/\lambda$ are allowed to grow simultaneously.

In case~(1), which arises in the context of studies of scattering by
randomly rough surfaces such as the Gaussian surfaces considered in
Section~\ref{sec_gaussian}, the cost of the algorithm grows
at most linearly with the number of unknowns---regardless of the incidence
angle, and including near grazing incidences. This favorable behavior
stems from the decay experienced by the shifted Green function $G_j$
used in~\eqref{shifted_periodic_green_function} as $d/\lambda$ grows
while keeping a constant height $H/\lambda$
(cf.~\eqref{shifted_green_space}
and~\cite[Sec. 5.4]{BrunoDelourme}). As a result of this decay, the
number $n_\mathrm{per}$ of terms necessary to obtain a prescribed error tolerance in
the summation of~\eqref{shifted_periodic_green_function} decreases as
$d$ grows.  In case~(2), on the other hand, the computational cost is
generally observed to range from $O(N)$ up to $O(N^\frac{3}{2})$, and
it can even reach $O(N^2)$ for extreme geometries.

The cost of the overall algorithm can be affected significantly by the
value selected for the shift-parameter $h$ (or, rather, of the
dimensionless parameter $h/\lambda$). On one hand, this parameter
controls the rate of convergence of the spatial series for the shifted
Green function: smaller values of $h/\lambda$ result in faster
convergence of this series. On the other hand, however, use of very
small values of $h/\lambda$ does give rise to certain ill-conditioning
difficulties (which, for geometric reasons, become more and more
pronounced as the grating-depths increase~\cite{BrunoDelourme}). In
particular, since, for a fixed $h/\lambda$ value, the distance between
the scattering surface and the first shifted source decreases as the
depth of the surface is increased, to avoid ill-conditioned-related
accuracy losses it becomes necessary to use larger and larger values
of $h/\lambda$ as the surface height grows. The selection of such
larger $h/\lambda$ values, in turn, requires use of increasingly
higher number of periods for the summation of the spatial periodic
Green function to maintain accuracy. For the test cases considered in
this paper, values of $h/\lambda$ in the range $\frac{1}{3}\leq
h/\lambda \leq 1$ were generally used.  For even steeper gratings,
larger upper bounds must be utilized in order to maintain a given
accuracy tolerance.

In any case, examination of the numerical results presented in what
follows does indicate that, for highly challenging scattering
configurations of the types that arise in a wide range of
applications, the accelerated solver introduced in this paper provides
significant performance improvements over the previous state of the
art: the proposed solver is often hundreds of times faster and beyond,
and significantly more accurate, than other available approaches. And,
importantly, it is applicable to Wood anomaly configurations, and it
is extensible to the three-dimensional case while maintaining a full
Wood-anomaly capability~\cite{3DGratingsFast}.

\subsection{Convergence}

In order to assess the convergence rate of the proposed algorithm, we
consider the problem of scattering of an incident plane-wave at a
fixed incidence angle $\theta = 45^\circ$ by the composite
surface~\cite{BrunoHaslamJOSA}
$$ f(x) = -\frac 14 \left( \sin(x) + \frac 12 \sin(2x) + \frac 13 \sin(3x) + \frac 14 \sin(4x) \right), \qquad x\in (0,2\pi) $$
depicted in Figure~\ref{fig_conv}, whose peak to trough height $H =
\max(f) - \min(f)$ equals $0.763$, and whose period $d$ equals
$2\pi$. For this test we consider two slightly different wavenumbers,
namely, the non-Wood wavenumber $\mathrm{k}=20$, for which we have
$\frac{H}{\lambda} = 2.43$ and $\frac{d}{\lambda} = 20$, and the Wood
wavenumber $\mathrm{k} = 6 (1-\sin(\theta))^{-1} \approx 20.4852...$
for which the $\frac{H}{\lambda}$ and $\frac{d}{\lambda}$ ratios are
slightly larger. Table~\ref{table_composite_conv} presents results of
convergence studies for these two test configurations, using the
unshifted Green function ($j=0$) for the non-Wood cases, and relying,
for the Wood cases, on the shifted Green function with shift-parameter
values $j=8$ and $h=0.16 \approx \lambda/2$. In both cases the
accelerator parameters $L=\lambda$, $n_\mathrm{eq}=10$ and
$n_\mathrm{plw}=35$ and $n_\mathrm{coll}=200$ were used. This table
displays the calculated values $\varepsilon$ of the energy-balance
error~\eqref{energy_error} as well as the error $\tilde\varepsilon$
defined as the maximum for $n\in U$ of the errors in each one of the
scattering efficiencies $e_n$ (Section~\ref{sec:scattering_prb}). (The
quantities $\tilde\varepsilon$ in Table~\ref{table_composite_conv}
were evaluated by comparison with reference values obtained using
large values of $N$ and $n_\mathrm{per}$.)

\begin{table}[ht!] 
\centering 
\caption{Convergence in a simple composite surface for Wood and non-Wood cases. \label{table_composite_conv}}
\resizebox{0.8\textwidth}{!}{\begin{tabular}{ | c | c | c | c | c | c | c | c | } 
\hline 
\multicolumn{2}{|c|}{}  & \multicolumn{3}{|c|}{ $\mathrm{k}=20$ (non-Wood)}  & \multicolumn{3}{|c|}{ $\mathrm{k}=20.4852...^a$ (Wood Anomaly)}  \\ 
\hline 
 $N$ & $n_\mathrm{per}$ & Total time & $\varepsilon$ & $\tilde\varepsilon$ & Total time & $\varepsilon$ & $\tilde\varepsilon$ \\ 
\hline 
   100     &     50     &   0.09  sec & 5.1e-03     &  1.3e-03     &   0.67  sec & 5.9e-02     &  2.2e-02     \\ 
   150     &     75     &   0.09  sec & 1.0e-05     &  4.2e-05     &   0.84  sec & 9.0e-04     &  2.8e-04     \\ 
   200     &    100     &   0.10  sec & 4.9e-06     &  4.2e-05     &   1.02  sec & 3.4e-05     &  7.0e-05     \\
   300     &    150     &   0.13  sec & 1.2e-06     &  2.3e-06     &   1.39  sec & 2.4e-06     &  9.0e-06     \\
   400     &    200     &   0.16  sec & 4.1e-07     &  1.8e-07     &   1.77  sec & 1.6e-07     &  6.1e-07     \\
   600     &    300     &   0.26  sec & 1.1e-08     &  4.9e-09     &   2.57  sec & 1.3e-07     &  2.6e-07     \\
   800     &    400     &   0.36  sec & 2.2e-11     &  3.1e-10     &   3.40  sec & 6.7e-08     &  4.8e-08     \\
\hline 
\end{tabular} 
}
\smallskip

\footnotesize
$^a$The exact value of the Wood-Anomaly frequency $\mathrm{k} = 6 (1-\sin(45^\circ))^{-1}$ was used.
\end{table} 

Table~\ref{table_composite_conv} demonstrates the high-order
convergence and efficiency enjoyed by the proposed algorithm, even for
Wood configurations for which the classical Green function is not even
defined. Concerning accuracy, we see that a mere doubling of the
number of discretization points and the number of terms used for
summation of the shifted Green function suffices to produce
significant improvements in the solution error. Additionally, an
increase in computing costs by a factor of five (from the first to the
last row in the table) suffices to increase the solution accuracy by
six additional digits. And, concerning efficiency, the table displays
computing times that grow in a slower-than-linear fashion as the
discretizations parameters $N$ and $n_\mathrm{per}$ are increased. (As
indicated above, the accelerator parameter $n_\mathrm{eq}=10$ is kept
fixed: the resulting rather-coarse discretization suffices to produce
all accuracies displayed in Table~\ref{table_composite_conv}.)

\begin{figure}[ht!]
\centering
\includegraphics[scale=0.5]{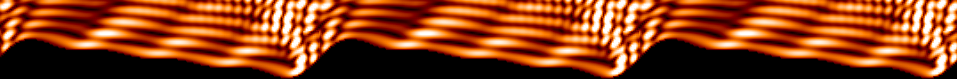}
\caption{Depiction of the solution of the Wood-anomaly problem
  considered in Table~\ref{table_composite_conv}. This solution
  resulted from a 0.9 sec. computation, which included the evaluation
  of the scattered field displayed. \label{fig_conv}}
\end{figure} 

\subsection{Sinusoidal Gratings\label{sinusoidal_gratings}}

In order to illustrate the performance of the proposed solver for a
wide range of problem parameters we consider a Littrow mount
configuration of order $-1$ (the $n=-1$ diffracted mode is
backscattered~\cite{Maystre}), with incidence angle $\theta$ given by
$\sin(\theta)=\frac 13$, for the sinusoidal surface
$$ f(x) = \frac{H}{2} \sin( 2\pi x / d ), \qquad x \in (0,d),$$
and with $H=\frac{d}{4}$ (Tables~\ref{table_away_wood_1}
and~\ref{table_wood_1}), $H=\frac{d}{2}$
(Tables~\ref{table_away_wood_2} and~\ref{table_wood_2}) and $H=d$
(Tables~\ref{table_away_wood_3} and~\ref{table_wood_3}). In the Wood
cases the wavenumber $\mathrm{k}$ varies from the first Wood frequency
($\mathrm{k}=1.5$) up to the sixth one ($\mathrm{k}=9$). As in the
previous section, the accelerator parameters $L=\lambda$,
$n_\mathrm{eq}=10$, $n_\mathrm{plw}=35$ and $n_\mathrm{coll}=200$ were
used in all
cases. 
Tables~\ref{table_away_wood_1},~\ref{table_away_wood_2}
and~\ref{table_away_wood_3} (resp. Tables~\ref{table_wood_1},
\ref{table_wood_2} and~\ref{table_wood_3}) correspond to non-Wood
(resp. Wood) configurations. The first row in each one of these tables
corresponds to test problems considered in~\cite[Tables
3-7]{BrunoDelourme}. 

The columns ``Iter. time'' and ``\# Iters.'' display the computing time required by each full solver iteration and the total number of iterations required to reach the energy balance tolerance $\varepsilon$. The columns ``$G^{\textit{qper}}_0$ eval.''  and ``Init. time'', in turn,  list initialization times as described in Remark~\ref{init_time_desc}.

\begin{table}[ht!] 
  \centering 
  \caption{Sinusoidal scatterer data for increasingly higher non-Wood frequencies; $H=\frac{d}{4}$, $j=0$. }
 \label{table_away_wood_1}   
  \resizebox{0.9\textwidth}{!}{\begin{tabular}{ | c | c | c | c | c | c | c | c | c | c | } 
\hline 
$H/\lambda$ & $d/\lambda$ & $N$ & $n_\mathrm{per}$ & $G^{\textit{qper}}_0$ eval. & Init. time & Iter. time & \# Iters. & Total time & $\varepsilon$  \\   
\hline 
  0.25     &   1.00     &     48     &    110     &   0.01  sec &   0.02  sec & 2.9e-04  sec &      7     &   0.02  sec &  1.7e-08    \\  
  0.62     &   2.50     &     76     &    110     &   0.01  sec &   0.03  sec & 5.8e-04  sec &     10     &   0.04  sec &  3.1e-08    \\  
  1.00     &   4.00     &    120     &    110     &   0.01  sec &   0.04  sec & 1.2e-03  sec &     12     &   0.06  sec &  7.7e-08    \\  
  1.38     &   5.50     &    166     &    110     &   0.01  sec &   0.11  sec & 1.0e-03  sec &     13     &   0.13  sec &  2.1e-08    \\  
  1.75     &   7.00     &    210     &    110     &   0.02  sec &   0.10  sec & 1.1e-03  sec &     14     &   0.12  sec &  2.1e-08    \\  
  2.12     &   8.50     &    256     &    110     &   0.02  sec &   0.08  sec & 2.2e-03  sec &     15     &   0.12  sec &  1.8e-09    \\  
\hline  
\end{tabular} 
}

\end{table}
\begin{remark}\label{init_time_desc}
  In Tables~\ref{table_away_wood_1} and subsequent, the columns
  ``Init. time'' display the total initialization times---that is, the
  times required in each case by Algorithm 2 in
  Section~\ref{overall_fast_solv}. This time includes, in particular,
  the separately-listed ``$G^{\textit{qper}}_0$ eval.'' time, which is
  the time required for the evaluation of all necessary values of the
  quasi-periodic Green function.
\end{remark}

\begin{table}[ht!] 
\centering 
\caption{Sinusoidal scatterer data for increasingly higher non-Wood frequencies; $H=\frac{d}{2}$, $j=0$.}   
\resizebox{0.9\textwidth}{!}{\begin{tabular}{ | c | c | c | c | c | c | c | c | c | c | } 
\hline 
$H/\lambda$ & $d/\lambda$ & $N$ & $n_\mathrm{per}$ & $G^{\textit{qper}}_0$ eval. & Init. time & Iter. time & \# Iters. & Total time & $\varepsilon$  \\   
\hline 
  0.50     &   1.00     &     64     &    120     &   0.01  sec &   0.03  sec & 6.2e-04  sec &      8     &   0.03  sec &  5.9e-08    \\  
  1.25     &   2.50     &    106     &    120     &   0.01  sec &   0.07  sec & 5.6e-04  sec &     13     &   0.07  sec &  6.1e-08    \\  
  2.00     &   4.00     &    168     &    120     &   0.01  sec &   0.08  sec & 1.4e-03  sec &     18     &   0.10  sec &  3.8e-09    \\  
  2.75     &   5.50     &    232     &    120     &   0.01  sec &   0.11  sec & 2.1e-03  sec &     21     &   0.15  sec &  6.3e-09    \\  
  3.50     &   7.00     &    294     &    120     &   0.02  sec &   0.11  sec & 2.5e-03  sec &     23     &   0.17  sec &  1.4e-09    \\  
  4.25     &   8.50     &    358     &    120     &   0.02  sec &   0.14  sec & 3.1e-03  sec &     26     &   0.22  sec &  3.3e-09    \\  
\hline  
\end{tabular} 
}
 \label{table_away_wood_2}   
\end{table}  

\begin{table}[ht!] 
\centering 
\caption{Sinusoidal scatterer data for increasingly higher non-Wood frequencies; $H=d$, $j=0$.}   
\resizebox{0.9\textwidth}{!}{\begin{tabular}{ | c | c | c | c | c | c | c | c | c | c | } 
\hline 
$H/\lambda$ & $d/\lambda$ & $N$ & $n_\mathrm{per}$ & $G^{\textit{qper}}_0$ eval. & Init. time & Iter. time & \# Iters. & Total time & $\varepsilon$  \\   
\hline 
  1.00     &   1.00     &     76     &    150     &   0.01  sec &   0.04  sec & 4.1e-04  sec &     12     &   0.05  sec &  2.2e-08    \\  
  2.50     &   2.50     &    126     &    150     &   0.01  sec &   0.05  sec & 1.2e-03  sec &     18     &   0.08  sec &  2.2e-08    \\  
  4.00     &   4.00     &    200     &    150     &   0.02  sec &   0.07  sec & 2.3e-03  sec &     26     &   0.13  sec &  2.0e-08    \\  
  5.50     &   5.50     &    276     &    150     &   0.02  sec &   0.15  sec & 3.7e-03  sec &     32     &   0.27  sec &  2.7e-09    \\  
  7.00     &   7.00     &    350     &    150     &   0.02  sec &   0.22  sec & 8.1e-03  sec &     39     &   0.54  sec &  5.6e-09    \\  
  8.50     &   8.50     &    426     &    150     &   0.03  sec &   0.41  sec & 9.3e-03  sec &     46     &   0.84  sec &  2.2e-09    \\  
\hline  
\end{tabular} 
}
 \label{table_away_wood_3}   
\end{table}

\begin{table}[ht!] 
\centering 
 \caption{Sinusoidal scatterer data for increasingly higher Wood frequencies. $H=\frac{d}{4}$}   
\resizebox{0.9\textwidth}{!}{\begin{tabular}{ | c | c | c | c | c | c | c | c | c | c | c | } 
\hline 
$H/\lambda$ & $d/\lambda$ & $N$    &   $h/\lambda$   & $n_\mathrm{per}$  & $G^{\textit{qper}}_8$ eval. & Init. time & Iter. time & \# Iters. & Total time & $\varepsilon$  \\   
\hline                                            
  0.38     &   1.50     &     46   &   0.43          &     50     &   0.03  sec &   0.05  sec & 2.1e-04  sec &     10     &   0.05  sec &  4.5e-08    \\  
  0.75     &   3.00     &     90   &   0.43          &     50     &   0.05  sec &   0.09  sec & 4.3e-04  sec &     17     &   0.10  sec &  7.8e-08    \\  
  1.12     &   4.50     &    136   &   0.43          &     50     &   0.09  sec &   0.15  sec & 6.9e-04  sec &     23     &   0.16  sec &  8.3e-08    \\  
  1.50     &   6.00     &    180   &   0.43          &     50     &   0.12  sec &   0.17  sec & 1.1e-03  sec &     30     &   0.20  sec &  9.0e-08    \\  
  1.88     &   7.50     &    226   &   0.48          &     50     &   0.13  sec &   0.21  sec & 1.0e-03  sec &     34     &   0.25  sec &  3.1e-08    \\  
  2.25     &   9.00     &    270   &   0.53          &     50     &   0.21  sec &   0.29  sec & 2.3e-03  sec &     38     &   0.37  sec &  5.9e-08    \\  
\hline  
\end{tabular} 
}
 \label{table_wood_1}   
\end{table}  

\begin{table}[ht!] 
\centering 
\caption{Sinusoidal scatterer data for increasingly higher Wood frequencies. $H=\frac{d}{2}$}   
\resizebox{0.9\textwidth}{!}{\begin{tabular}{ | c | c | c | c | c | c | c | c | c | c | c | } 
\hline 
$H/\lambda$ & $d/\lambda$ &  $N$  & $h/\lambda$     &     $n_\mathrm{per}$ & $G^{\textit{qper}}_8$ eval. & Init. time & Iter. time & \# Iters. & Total time & $\varepsilon$  \\   
\hline                                           
  0.75     &   1.50     &     90  &   0.36          &    200     &   0.09  sec &   0.14  sec & 4.0e-04  sec &     15     &   0.15  sec &  2.5e-08    \\  
  1.50     &   3.00     &    180  &   0.48          &    200     &   0.29  sec &   0.38  sec & 1.2e-03  sec &     23     &   0.41  sec &  7.6e-08    \\  
  2.25     &   4.50     &    270  &   0.69          &    400     &   0.68  sec &   0.86  sec & 1.9e-03  sec &     26     &   0.91  sec &  3.0e-08    \\  
  3.00     &   6.00     &    360  &   0.69          &    400     &   1.13  sec &   1.39  sec & 3.0e-03  sec &     34     &   1.49  sec &  3.3e-08    \\  
  3.75     &   7.50     &    450  &   0.74          &    600     &   1.92  sec &   2.22  sec & 2.6e-03  sec &     40     &   2.33  sec &  3.9e-08    \\  
  4.50     &   9.00     &    540  &   0.77          &    600     &   3.24  sec &   3.59  sec & 5.1e-03  sec &     46     &   3.83  sec &  2.3e-08    \\  
\hline  
\end{tabular} 
}
 \label{table_wood_2}   
\end{table}  

\begin{table}[ht!] 
\centering 
\caption{Sinusoidal scatterer data for increasingly higher Wood frequencies. $H=d$}   
\resizebox{0.9\textwidth}{!}{\begin{tabular}{ | c | c | c | c | c | c | c | c | c | c | c | } 
\hline 
$H/\lambda$ & $d/\lambda$ & $N$    & $h/\lambda$   &         $n_\mathrm{per}$ & $G^{\textit{qper}}_8$ eval. & Init. time & Iter. time & \# Iters. & Total time & $\varepsilon$  \\   
\hline                                           
  1.50     &   1.50     &    200   &   0.36        &    400     &   0.18  sec &   0.50  sec & 7.6e-04  sec &     27     &   0.52  sec &  2.8e-09    \\  
  3.00     &   3.00     &    400   &   0.57        &    650     &   0.83  sec &   1.48  sec & 2.0e-03  sec &     37     &   1.56  sec &  1.7e-08    \\  
  4.50     &   4.50     &    600   &   0.79        &   1000     &   1.87  sec &   3.20  sec & 3.1e-03  sec &     46     &   3.34  sec &  1.5e-08    \\  
  6.00     &   6.00     &    800   &   0.86        &   1500     &   4.86  sec &   6.52  sec & 9.6e-03  sec &     59     &   7.09  sec &  6.4e-08    \\  
  7.50     &   7.50     &   1000   &   0.90        &   2000     &   8.29  sec &  10.67  sec & 9.2e-03  sec &     74     &  11.35  sec &  5.8e-07    \\  
  9.00     &   9.00     &   1200   &   0.86        &   2500     &  17.22  sec &  19.81  sec & 9.8e-03  sec &     88     &  20.68  sec &  3.2e-08    \\  
\hline  
\end{tabular}
} 
 \label{table_wood_3}   
\end{table}  

The non-Wood examples considered in
Tables~\ref{table_away_wood_1},~\ref{table_away_wood_2}
and~\ref{table_away_wood_3} demonstrate the performance of the
proposed accelerated solver in absence of Wood anomalies: these
results extend corresponding data tables presented in the recent
reference~\cite{BrunoDelourme}, with better than single precision
accuracy, to problems that are up to eight times higher in frequency
and depth in comparable sub-second, single-core computing times (cf.
Tables~2,~3 and~4 in~\cite{BrunoDelourme}). High accuracy and speed
are also demonstrated in the Wood-anomaly cases considered in
Tables~\ref{table_wood_1},~\ref{table_wood_2}
and~\ref{table_wood_3}. With exception of the first row in each one of
these tables, for which comparable performance was demonstrated
in~\cite{BrunoDelourme}, none of these problems had been previously
treated in the literature. These tables demonstrate that better than
single precision accuracy is again produced by the proposed methods at
the expense of modest computing costs.

Increases by factors of 2.5 to 25 are observed in the ``Total time''
columns of the Wood-anomaly tables in this section relative to the
corresponding columns in the non-Wood tables, with cost-factor
increases that grow as $\frac{H}{d}$ and/or $\frac{H}{\lambda}$
grow. The cost increases at Wood frequencies, which can be tracked
down directly to the cost required of evaluation of the shifted Green
function, are most marked for deep gratings---which, as discussed in
Section~\ref{chap2_comp_cost}, require use of adequately enlarged
values of the shift parameter $h$ to avoid near singularity and ill
conditioning, and which therefore require use of larger numbers
$n_\mathrm{per}$ of terms for the summation of the shifted
quasi-periodic Green function $\tilde G^{\textit{qper}}_j$.


\subsection{Large random rough surfaces under
  near-grazing incidence\label{sec_gaussian}}
This section demostrates the character of the proposed algorithm in
the context of randomly rough Gaussian surfaces under near-grazing
illumination.  At exactly grazing incidence, $\theta = 90^\circ$, the
zero-th efficiency becomes a Wood anomaly---a challenge which
underlies the significant difficulties classically found in the
solution of {\em near grazing} periodic rough-surface scattering
problems.

Various techniques~\cite{johnson1998grazing,saillard2011} based on
tapering of either the incident field, or the surface, or both, have
been proposed to avoid the nonphysical edge diffraction which arises
as an infinite random surfaces is truncated to a bounded computational
domain. Unfortunately, the modeling errors introduced by this
approximation are strongly dependent on the incidence angle and the
size of the truncated
section~\cite{johnson1998grazing,saillard2011}. Consideration of
periodic surfaces~\cite{chenwest1995} provides an alternative that
does not suffer from this difficulty. However, periodic-surface
approaches have only occasionally been pursued in the context of
random surfaces, on the basis that while~\cite{johnson1998grazing}
``periodic surfaces [allow use of] plane wave incident fields without
angular resolution problems [...] these techniques do not simultaneously
model a full range of ocean length scales for microwave and higher
frequencies''. Thus, the contribution~\cite{johnson1998grazing}
proposes use of a taper---an approach which has been influential in
the subsequent literature~\cite{saillard2011}. As demonstrated in this
section, the proposed periodic-surface solvers can tackle wide ranges
of length-scales, thus eliminating the disadvantages of the periodic
simulation method while maintaining its main strength: direct
simulation of an unbounded randomly rough surface.

The character of the proposed solvers in the random-surface context is
demonstrated by means of a range of challenging numerical
examples. Throughout this section surface ``heights'' are quantified
in terms of the surface's root-mean-square height (rms). For
definiteness, all test cases concern randomly-rough Gaussian
surfaces~\cite[p. 124]{Tsang_vol2} with correlation length equal to
the electromagnetic wavelength $\lambda$; examples for various
period-to-wavelength and height-to-wavelength ratios are used to
demonstrate the computing-time scaling of the algorithm.  Equispaced
meshes of meshsize $\Delta x = \lambda/10$ (Section~\ref{quadrature})
were used for all the examples considered in this section.

\begin{table}[ht!]
\centering
\caption{Gaussian surface with $\theta=89.9^\circ$, $H=\frac{\lambda}{2}$ mean rms. \label{grazing}}
\resizebox{0.8\textwidth}{!}{\begin{tabular}{| c | c | c | c | c | c | c | c |} 
\hline 
$d$/$\lambda$ & $n_\mathrm{per}$ & $G^{\textit{qper}}_8$ eval. & Init. time &  Iter. time & \# Iters. & Total time & $\varepsilon$  \\ 
\hline 
    25     &   1600        &   4.16 sec  &   6.89 sec  &   3.5e-03  sec     &    103     &   7.39 sec & 1.8e-08   \\        
    50     &    800        &   3.76 sec  &   6.66 sec  &   7.1e-03  sec     &    209     &   8.03 sec & 2.5e-07   \\        
   100     &    400        &   3.62 sec  &   8.77 sec  &   1.3e-02  sec     &    360     &  13.81 sec & 3.2e-08   \\        
   200     &    200        &   3.70 sec  &  14.56 sec  &   2.6e-02  sec     &    680     &  33.10 sec & 4.6e-08   \\        
   300     &    133        &   4.06 sec  &  20.13 sec  &   3.8e-02  sec     &    973     &  57.93 sec & 3.2e-08   \\        
   400     &    100        &   4.48 sec  &  26.77 sec  &   5.5e-02  sec     &   1242     &  96.02 sec & 4.6e-08   \\        
\hline 
\end{tabular} 
}
\end{table}

Table~\ref{grazing} presents computing times and accuracies for
problems of scattering by Gaussian surfaces of rms-height equal to
$\lambda/2$ under close-to-grazing incidence $\theta=89.9^\circ$. The
data displayed in this table demonstrates uniform accuracy, with fixed
meshsize, for periods going from twenty-five to four-hundred
wavelengths in size. Certain useful characteristics of the algorithm
may be gleaned from this table. On one hand, the ``time'' columns in
the table show that, as indicated in Section~\ref{sec_intro} and
discussed in Section~\ref{chap2_comp_cost}, the computing costs for a
fixed accuracy grow at most linearly with the surface period
$d/\lambda$. The ``$G^{\textit{qper}}_8$ eval.'' data, in turn, shows
that the cost of evaluation of the shifted Green function
$G^{\textit{qper}}_j$ with $j=8$ remains essentially constant as the
size of the surface grows---and that, therefore, the Green-function
cost becomes negligible, when compared to the total cost, for
sufficiently large surfaces. The $\varepsilon$ error column
demonstrates the high accuracy of the method.

\begin{remark} The ``constant-cost'' observed for the computation of
  $G^{\textit{qper}}_8$ in Table~\ref{grazing} can be understood as
  follows. As noted in section~\ref{hybrid_spectral}, the efficiency
  of the spectral series is inversely proportional to parameter
  $\frac{\delta}{d}$, where $\delta$ is the distance from $Y$ to the
  set of polar points $\{ -mh, 0\le m \le j \}$. As the period $d$
  grows the quotient $\frac{\delta}{d}$ decreases, and, therefore, the
  trade-off in the hybrid strategy increasingly favors the use of the
  spatial series---which as demonstrated by the $n_\mathrm{per}$
  column in Table~\ref{grazing}, requires smaller and smaller values
  of $n_\mathrm{per}$ as the period is increased, to meet a given
  error tolerance.
\end{remark}

Figure~\ref{fig:gauss} displays scattered fields produced by
increasingly {\em larger and steeper} Gaussian surfaces under
$89^\circ$ near-grazing incidence. The $\varepsilon$ error is in all
cases of the order of $10^{-9}$, and the computing times reported in
the figure caption include the computation of the displayed near
field.

\begin{figure}[ht!]
\includegraphics[width=500pt, height=30pt]{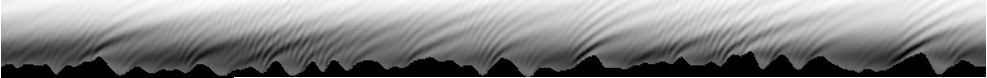}
 
\includegraphics[width=500pt, height=30pt]{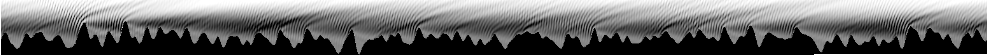}
 
\includegraphics[width=500pt, height=30pt]{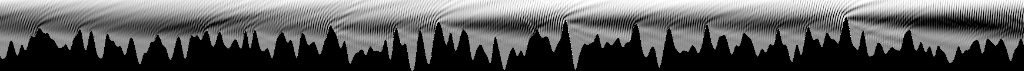}

\caption{Gaussian rough surfaces under $\theta=89^\circ$ incidence,
  with simulation errors $\varepsilon < 10^{-8}$ in all cases. Top:
  $d=100 \lambda$, $H=\frac{\lambda}{2}$ mean rms ($2.6\lambda$
  peak-to-trough). Center: $d=200 \lambda$, $H=\lambda$ mean rms
  ($6.7\lambda$ peak-to-trough).  Bottom: $d=1000 \lambda$ (fragment),
  $H=2\lambda$ mean rms ($14.3\lambda$ peak-to-trough). Computing time
  (including near field evaluation) is 22.3 sec., 62.9 sec. and 830
  sec. respectively.\label{fig:gauss}}
\end{figure}

\subsection{Comparison with~\cite{BrunoHaslamJOSA} for some ``extreme'' problems\label{mike}}

A number of fast and accurate solutions were provided
in~\cite{BrunoHaslamJOSA} for highly-challenging grating-scattering
problems (in configurations away from Wood Anomalies); relevant
performance comparisons with results in that contribution are
presented in what follows. While the results of~\cite{BrunoHaslamJOSA}
ensure accuracies of the order of ten to twelve digits, the solver
introduced in the present paper was restricted, for definiteness, to
accuracies of the order of single-precision. Fortunately, however,
Table~8 in~\cite{BrunoHaslamJOSA} presents a convergence study for a
problem of scattering by a composite surface. That table shows that
the method~\cite{BrunoHaslamJOSA} requires 85 seconds to reach single
precision accuracy for this problem; the present approach, in
contrast, reaches the same precision for the same problem in just 1.8
seconds---including the evaluation of the near-field displayed in
Figure~\ref{large_composite}.
\begin{remark}\label{times_mike} 
  Higher accuracies can be produced by the present approach at
  moderate additional computational expense. In turn, results in
  Table~8 in~\cite{BrunoHaslamJOSA} show that, for example, a
  reduction in accuracy from fourteen digits to single precision only
  produces a relatively small reduction in computing time---from 98
  seconds to 85 seconds. This is a consequence, of course, of the
  high-order convergence of the method~\cite{BrunoHaslamJOSA}.
\end{remark}
\begin{figure}[ht!]
\centering
\includegraphics[scale=0.5]{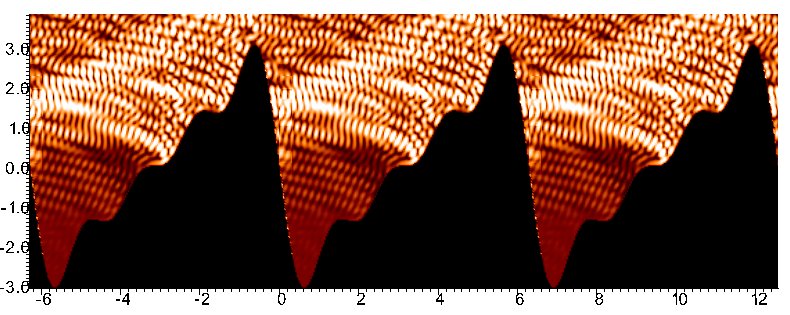} 
\caption{Depiction of the solution of the considered in Table 8
  of~\cite{BrunoHaslamJOSA}. This solution resulted from a 1.8
  sec. computation, which included the evaluation of the scattered
  field displayed.}
\label{large_composite}
\end{figure}

As an additional example we consider Table 5
in~\cite{BrunoHaslamJOSA}. That table presents results for extremely
deep sinusoidal gratings with $\lambda=0.05$ and incidence angle
$\theta=70^\circ$. The corresponding accuracies and computing times
produced for those configurations by the present solvers are presented
in Table~\ref{deep_gratings}. Comparison of the tabulated data shows
significant improvements in computing times, by factors of 12 to 25,
at the expense of a few digits of accuracy; see
Remark~\ref{times_mike}.

\begin{table}[ht!] 
\centering 
\caption{Increasingly deep gratings with a fixed period, and with incidence angle $\theta=70^\circ$.}
\resizebox{0.8\textwidth}{!}{\begin{tabular}{ | c | c | c | c | c | c | c | c | c | } 
\hline 
$h/\lambda$ & $d/\lambda$ & $N$ & $G^{\textit{qper}}_0$ eval. &  Init. time & Iter. time & \# Iters. & Total time & $\varepsilon$  \\ 
\hline 
   160     &     20     &    800     &   0.74  sec &   1.59  sec &   0.08  sec &    633     &   0.84  min &  5.9e-08    \\  
   320     &     20     &   1600     &   1.01  sec &   3.31  sec &   0.15  sec &   1260     &   3.30  min &  5.3e-08    \\  
   480     &     20     &   2400     &   1.28  sec &   4.73  sec &   0.26  sec &   1881     &   8.21  min &  2.6e-08    \\  
   640     &     20     &   3200     &   1.59  sec &   8.89  sec &   0.35  sec &   2507     &  14.88  min &  6.1e-08    \\  
   800     &     20     &   4000     &   1.97  sec &   9.96  sec &   0.43  sec &   3148     &  22.83  min &  8.0e-08    \\  
\hline 
\end{tabular} 
}
\label{deep_gratings}
\end{table} 

Table 7 in~\cite{BrunoHaslamJOSA}, finally, considers increasingly
high frequencies while maintaining the other problem parameters fixed:
$\theta = 45^\circ$, $d=1$, $h=2$.  A similar picture emerges in this
case: the method~\cite{BrunoHaslamJOSA} solves problems with
accuracies of the order of 13 to 16 digits, at computing times that
are larger than those displayed in Table~\ref{high_freq} by factors of
10 to 18.
                                                           
\begin{table}[ht!] 
\centering 
\caption{Increasingly high frequencies, with $\theta = 45^\circ$, $d=1$, $h=2$}
 \resizebox{0.8\textwidth}{!}{\begin{tabular}{ | c | c | c | c | c | c | c | c | c | } 
\hline 
$h/\lambda$ & $d/\lambda$ & $N$ & $G^{\textit{qper}}_0$ eval. &  Init. time & Iter. time & \# Iters. & Total time & $\varepsilon$  \\ 
\hline 
    20     &     10     &    200     &   0.55  sec &   0.75  sec &   0.01  sec &     92     &   1.62  sec &  4.1e-09    \\  
    40     &     20     &    400     &   1.09  sec &   1.51  sec &   0.02  sec &    167     &   5.02  sec &  1.7e-08    \\  
   200     &    100     &   2000     &  11.57  sec &  13.64  sec &   0.25  sec &    477     & 133.67  sec &  3.8e-11    \\  
   400     &    200     &   4000     & 122.78  sec & 128.25  sec &   1.00  sec &    698     & 824.82  sec &  2.4e-09    \\  
\hline 
\end{tabular} 
}
\label{high_freq}
\end{table}

\section{Conclusions\label{concl}}

The periodic-scattering solver introduced in this paper provides the
first accelerated solver of high-order of accuracy for the solution of
problems of scattering by periodic surfaces up to and including Wood
frequencies. The algorithm relies on use of an accelerated shifted
Green function methodology which reduces operator evaluations to Fast
Fourier Transforms, and which, in particular, greatly reduces the
required number of evaluations of the shifted quasi-periodic Green
function. Significant additional acceleration is obtained by the
solver by means of an appropriate application of a dual
spectral/spatial approach for evaluation of the shifted Green
function---which exploits, when possible, the exponentially fast
convergence of the spectral series, and which relies on the
high-order-convergent shifted spatial series for points for which the
convergence of the spectral series deteriorates. The combined solver
is highly efficient: it enables fast and accurate solution of some of
the most challenging two-dimensional periodic scattering problems
arising in practice. A three-dimensional version of this approach has
been found equally effective, and will the subject of a subsequent
contribution.

\section*{Acknowledgment}
OB gratefully acknowledges support by NSF and AFOSR and DARPA through
contracts DMS-1411876 and FA9550-15-1-0043 and HR00111720035, and the
NSSEFF Vannevar Bush Fellowship under contract number
N00014-16-1-2808. MM work was supported from a PhD fellowship of
CONICET and the Bec.AR-Fullbright Argentine Presidential Fellowship in
Science and Technology.

\appendix
\section{Appendix: Convergence and error analysis\label{app_chap_2}}
An error analysis for the numerical method embodied in
equation~\eqref{discrete_eq} follows from the standard stability
result~\cite[Th.10.12]{Kress}. The following lemma establishes the
crucial new element necessary to produce a convergence estimate
specific to equation~\eqref{discrete_eq}, namely, an error estimate
for the combined smooth windowing and trapezoidal quadrature for the
operator $\tilde D_{\mathrm{reg}}$ (all other needed estimates can be
found in reference~\cite{Kress}). Throughout this section the
notations in Section~\ref{quadrature} are used together with the
shorthand $\vec{\mu} = [\mu(x_1),\dots,\mu(x_N)]$ for a given
quasi-periodic function $\mu$.

\begin{lemma}\label{lemma_app_1}
  Let $d>0$ and $\alpha \geq 0$, and let $\mu$ denote an infinitely
  differentiable $\alpha$-quasi-periodic function of quasi-period
  $[0,d]$. Then, $\tilde D_{\mathrm{reg}}^{A,\Delta x}[\vec{\mu}](x)$
  tends to $\tilde D_{\mathrm{reg}}[\mu](x)$, uniformly for $x\in
  [0,d]$, as $A \to \infty$ and $\Delta x \to 0$. More precisely, we
  have
  \begin{equation}\label{estimate_1}
    | \tilde D_{\mathrm{reg}}[\mu](x) - \tilde D_{\mathrm{reg}}^{A,\Delta x}[\vec{\mu}](x) | \le E_p (\Delta x)^p + C_q A^{-q}\quad (1\leq i\leq N),
  \end{equation}
  for all positive integers $p$, and with $q =
  \left\lfloor{\frac{j+1}{2}}\right\rfloor -\frac{1}{2}$ near Wood
  anomalies, and for all positive integers $q$ away from Wood-anomaly
  frequencies. Here $C_q$ and $E_p$ are constants that do not depend
  on either $A$ or $\Delta x$. We also have the error estimate
\begin{equation}\label{estimate_2}
| \tilde D_{\mathrm{reg}}[\mu](x) - \tilde D_{\mathrm{reg}}^{\Delta x}[\vec{\mu}](x) | \le E_p (\Delta x)^p \quad \mbox{for all } p\in\mathbb{N}.
\end{equation}
\end{lemma}
\begin{proof}
  Let us consider the triangle-inequality estimate
  \begin{equation}\label{lemma_ineq1}
    | \tilde D_{\mathrm{reg}}[\mu](x) - \tilde D_{\mathrm{reg}}^{A,\Delta x}[\vec{\mu}](x) | \le | \tilde D_{\mathrm{reg}}[\mu](x) - \tilde D_{\mathrm{reg}}^{A}[\mu](x) | + | \tilde D_{\mathrm{reg}}^A[\mu](x) - \tilde D_{\mathrm{reg}}^{A,\Delta x}[\vec{\mu}]| .
  \end{equation}
  The first term on the right hand side of this relation admits the
  bound
\begin{equation}\label{super-alg}
  | \tilde D_{\mathrm{reg}}[\mu](x) - \tilde D_{\mathrm{reg}}^{A}[\mu](x) | \le C_q A^{-q},
\end{equation}
for certain values of $q$, as indicated as follows.  For frequencies
$\mathrm{k}$ away from Wood anomalies, on one hand, the Green function
series converges at a superalgebraic rate as
$A\to\infty$~\cite{BrunoDelourme}, (faster than $A^{-q}$ for any
integer $q$), for all integers $j\geq 0$ (including the ``unshifted''
case $j=0$), and thus so does $\tilde D_{\mathrm{reg}}^{A}[\mu]$.  In
other words, away from Wood anomalies, the bound~\eqref{super-alg}
holds for all positive integers $q$. For frequencies $\mathrm{k}$ up to and
including Wood anomalies, on the other hand,
reference~\cite{BrunoDelourme} shows that for a given integer $j\geq
1$, the Green function series enjoys algebraic convergence, with
errors of the order of $A^{-q})$ with $q= (j-1)/2$ for $j$ even, and
with $q= j/2$ for $j$ odd. It follows that, up to an including Wood
anomalies, for a given $j\geq 1$ the bound~\eqref{super-alg} holds
with $q = \left\lfloor{\frac{j+1}{2}}\right\rfloor -\frac{1}{2}$.

Having obtained the estimate~\eqref{super-alg} for the first term on
the right-hand side of~\eqref{lemma_ineq1} under the various frequency
regimes, we now turn to the second term on that right-hand side. To
estimate this term, we first consider the smooth $2A$-periodic
function $F^{A,x}=F^{A,x}(x')$ (which, as indicated in
Remark~\ref{rem_trapezoidal_Dreg}, coincides with the integrand in
equation~\eqref{na_A}), and we show that the coefficients
\begin{equation}
  \label{four-coeffs}
  F^{A,x}_n = \frac{1}{2A}\int_{-A}^A F^{A,x}(x') e^{-\frac{\pi i}Anx'}dx'
\end{equation}
of the Fourier series
\begin{equation}
  \label{eq:Fourier_x}
F^{A,x}(x')=\sum_{n=-\infty}^\infty F^{A,x}_n e^{\frac{\pi i}Anx'} 
\end{equation}
converge to zero rapidly and uniformly in $A$ and $x$ as
$n\to\infty$. Indeed, using integration by parts $p$ times
in~\eqref{eq:Fourier_x} we see that
\begin{equation}
  \label{eq:fcoef_est}
  |F^{A,x}(x')|\leq C^{A,x}_p \left( \frac{A}{n}\right)^p
\end{equation}
where $C^{A,x}_p$ is an upper bound for the absolute value of the
product of $\pi^p$ and the $p$-th derivative of $F^{A,x}(x')$ with
respect to $x'$. But, considering the expression that defines
$F^{A,x}(x')$, namely, the integrand in~\eqref{na_A}, we see that the
$p$-th order derivative of $F^{A,x}(x')$ with respect to $x'$ is
bounded by a constant which does not depend on $A$ or $x$---since the
same is true of each of the four functions in~\eqref{na_A} whose
products equals $F^{A,x}$. We thus obtain, for each non-negative
integer $p$, the bound
\begin{equation}
  \label{eq:four-coefs_bd}
  |F^{A,x}_n|  < C_p \left( \frac{A}{n}\right)^p,
\end{equation}
where the constant $C_p$ depends on $p$ only.  Since $F^{A,x}(x')$ is
(a periodic extension of) the integrand in~\eqref{na_A}, we see that
$\tilde D_{\mathrm{reg}}^{A}[\mu]$ equals the zero-th order
coefficient of $F^{A,x}(x')$:
\begin{equation}
  \label{eq:zeroth-ord-est}
  \tilde D_{\mathrm{reg}}^{A}[\mu] = F^{A,x}_0.
\end{equation}
The discrete approximation
$D_{\mathrm{reg}}^{A,\Delta x}[\vec{\mu}](x)$ in~\eqref{DAdx_def}, in
turn, utilizes in the periodicity interval $[x-A,x+A]$ a number $N_A$
of discretization points that satisfies the relations
\begin{equation}
  \label{eq:A-dic-pts-numb}
\lfloor A/d\rfloor N \leq  N_A \leq \lceil A/d\rceil N
\end{equation}
where, for a real number $r$, $\lceil r \rceil$ (resp.
$\lfloor r\rfloor$) denotes the smallest integer larger than or
equal to $r$ (resp. the largest integer smaller than or equal to
$r$). For a given period $d$ we clearly have
\begin{equation}
  \label{eq:NA_numb}
  N_A = O\left(\frac{A}{N}\right).
\end{equation}

As is well known (and easily checked), the $N_A$-point discrete
trapezoidal-rule quadrature inherent in equation~\eqref{DAdx_def}
integrates correctly all the non-aliased harmonics in
equation~\eqref{eq:Fourier_x}, and it produces the value one for the
aliased harmonics. We thus obtain
\begin{equation}
  \label{eq:discr-est}
  D_{\mathrm{reg}}^{A,\Delta x}[\vec{\mu}] = \sum_{\ell = -\infty}^\infty F^{A,x}_{\ell N_A}.
\end{equation}
In view of~\eqref{eq:four-coefs_bd},~\eqref{eq:zeroth-ord-est} and
~\eqref{eq:discr-est} it follows that
\begin{equation}
  \label{eq:2nd-est}
  | \tilde D_{\mathrm{reg}}^A[\mu](x) - \tilde
  D_{\mathrm{reg}}^{A,\Delta x}[\vec{\mu}]| = \left |\sum_{\substack{\ell = -\infty \\ \ell \ne 0}}^\infty F^{A,x}_{\ell N_A}\right| \leq C_p \left( \frac {A}{N_A}\right)^p\sum_{\substack{\ell = -\infty \\ \ell \ne 0}}^\infty \ell^{-p}
\end{equation}
which, in view of~\eqref{eq:A-dic-pts-numb} and since
$\Delta x\sim 1/ N$, for $p\geq 2$ shows that
\begin{equation}
 | \tilde D_{\mathrm{reg}}^A[\mu](x) - \tilde D_{\mathrm{reg}}^{A,\Delta x}[\vec{\mu}](x) |  \le E_p (\Delta x)^p
\end{equation}
for some constant~$E_p$, as desired. The proof is now complete.
\end{proof}

\bibliographystyle{plain}
\bibliography{../all_references}

\end{document}